\documentclass{article}
\usepackage{arxiv}
\usepackage{cite}
\usepackage{amsmath,amssymb,amsfonts,amsthm,multirow}
\usepackage{algorithmic}
\usepackage{graphicx}
\usepackage{textcomp}
\usepackage{xcolor}
\usepackage[linesnumbered,ruled,vlined]{algorithm2e}
\usepackage{hyperref}
\usepackage{comment}
\usepackage{float}
\usepackage[caption = false,farskip=0pt]{subfig}

\newtheorem{assumption}{Assumption}
\newtheorem{lemma}{Lemma}
\newtheorem{theorem}{Theorem}
\newtheorem{remark}{Remark}
\newtheorem{cor}{Corollary}

\usepackage{mathtools}
\DeclarePairedDelimiter\floor{\lfloor}{\rfloor}

\def\BibTeX{{\rm B\kern-.05em{\sc i\kern-.025em b}\kern-.08em
    T\kern-.1667em\lower.7ex\hbox{E}\kern-.125emX}}

\begin{document}

\title{Order Determination for Tensor-valued Observations Using Data Augmentation
}

\author{
 Una Radoji\v{c}i\'c\\
{Vienna University of Technology}\\
\textit{una.radojicic@tuwien.ac.at}
\And
{ Niko Lietz\'en}\thanks{The work of NL was supported by the Academy of Finland (Grant 321968).}\\
{University of Turku}\\
Turku, Finland \\
\textit{niko.lietzen@utu.fi}
\And
{Klaus Nordhausen}\\
{University of Jyväskylä}\\
\textit{klaus.k.nordhausen@jyu.fi}
\And
{Joni Virta}\thanks{The work of JV was supported by the Academy of Finland (Grant 335077).}\\
{University of Turku}\\
\textit{joni.virta@utu.fi}
}
\maketitle

\begin{abstract}
Tensor-valued data benefits greatly from dimension reduction as the reduction in size is exponential in the number of modes. To achieve maximal reduction without loss in information, our objective in this work is to give an automated procedure for the optimal selection of the reduced dimensionality. Our approach combines a recently proposed data augmentation procedure with the higher-order singular value decomposition (HOSVD) in a tensorially natural way. We give theoretical guidelines on how to choose the tuning parameters and further inspect their influence in a simulation study. As our primary result, we show that the procedure consistently estimates the true latent dimensions under a noisy tensor model, both at the population and sample levels. Additionally, we propose a bootstrap-based alternative to the augmentation estimator. Simulations are used to demonstrate the estimation accuracy of the two methods under various settings.

\end{abstract}


\section{Introduction}


\subsection{Tensors and image data}

Tensors offer a particularly convenient framework for the modelling of different types of image data \cite{aja2009tensors,jouni2021}. For example, a collection of gray-scale images can be viewed as a sample of second-order tensors (i.e., matrices) where the tensor elements correspond to intensities of the pixels. A set of colour images can be represented as a sample of third-order tensors where the third mode has dimensionality equal to 3 and collects the colour information through RGB-values of the pixels. Still higher-dimensional tensors are obtained if we sample multiple images per subject (possibly in the form of a video).

The total number of elements in a tensor grows exponentially in its order. Therefore, as soon as our images have even a reasonably large resolution, the resulting data set consists of a huge number of variables. Consider, for instance, the \textit{butterfly} data set\footnote{The data set is freely available at \url{https://www.kaggle.com/datasets/gpiosenka/butterfly-images40-species}} that will serve as our running example throughout the paper; see the top row of Figure~\ref{fig:example} for five particular images in the data set. For simplicity, we use in this paper a subsample of the full data, consisting of $n = 882$ RGB images of butterflies of various species 
with the resolution $224 \times 224$. As such, the number of variables per image is $150528$. However, as any two neighbouring pixels in an image are typically highly correlated, the ``signal dimension'' of the butterfly data, and image data in general, is most likely much smaller than the total number of observed variables.

\begin{figure}[ht]
\centering
\subfloat{\includegraphics[width = 0.18\linewidth]{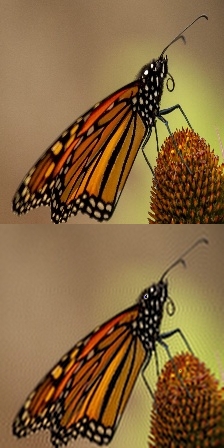}}
\subfloat{\includegraphics[width = 0.18\linewidth]{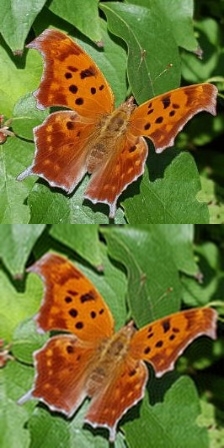}}
\subfloat{\includegraphics[width = 0.18\linewidth]{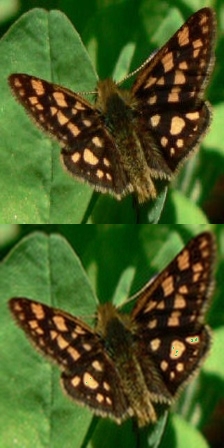}}
\subfloat{\includegraphics[width = 0.18\linewidth]{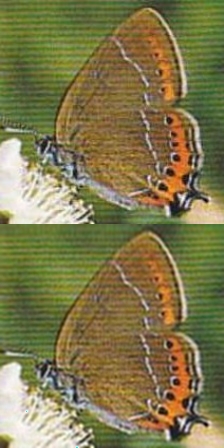}}
\subfloat{\includegraphics[width = 0.18\linewidth]{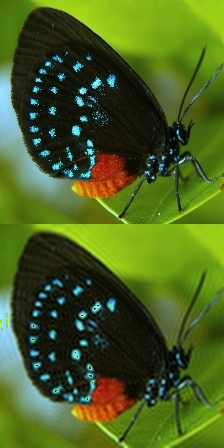}}
\caption{A collection of five images from the \textit{butterfly} data set (top row) and the corresponding reconstructed images using compressed, $103\times111\times3$-dimensional core tensors (bottom row).}
\label{fig:example}
\end{figure}

The standard approach to remove the redundant information and noise from the data is through low-rank decomposition where the images are approximated as products of low-rank tensors \cite{aja2009tensors,Inoue2016}. 
As with any dimension reduction method, a key problem in this procedure is the selection of the reduced rank/order/dimension (we use all three terms interchangeably in the sequel). We want to keep the order small not to incorporate noise in the decomposition but, at the same time, enough components should be included to ensure accurate capturing of the signal. As the main contribution of the current work, we develop a method for the automatic determination of the order, separately in each mode. Reconstructions of the five butterfly images based on the dimensionalities chosen by our method are shown in the bottom row of Figure \ref{fig:example}. For further details on this example, see Section \ref{sec:real_data_example}. We develop our method under a specific statistical model and show that asymptotically, when the number of images $n$ generated from the model grows without bounds, we are guaranteed to recover the true rank of the data.

Our main idea is based on the combination of a tensor decomposition known as \textit{higher-order singular valued decomposition} (HOSVD) \cite{Lathauwer2000} and a specific method of order determination known as \textit{predictor augmentation} \cite{LuoLi2021}. To set up the framework, the next two subsections briefly review the literature pertaining to these subjects and, afterwards, in Section \ref{sec:previous} we still highlight the contributions of this work in comparison to the existing literature. In the same section we also discuss the connection to the conference paper \cite{RadojicicLietzenNordhausenVirta2021}, of which this work is an extended version.

\subsection{Tensor decompositions}

The term tensor decomposition refers to a class of algorithms that approximate an input tensor as a sum or product of tensors/matrices of lower rank. The low-rank structure is usually thought to represent the signal/information of the original data tensor, whereas, if we manage to choose the rank appropriately, the difference between the original tensor and its low-rank approximation is pure noise. A comprehensive review of tensor decompositions is given by \cite{Kolda2009} and, for some recent uses of tensor decompositions in the context of image data, see \cite{Inoue2016,zhang2020robust, lou2019robust}.

Our decomposition of choice in this work is the higher-order singular value decomposition \cite{Lathauwer2000}. In HOSVD, an $m$th order tensor $\mathcal{A}\in\mathbb{R}^{p_1\times\cdots\times p_m} $ is approximated as $\mathcal{A} \approx \mathcal{B}\times_{i=1}^m\textbf{U}_i$, a multilinear product of the \textit{core tensor} $\mathcal{B}$ of a low dimensionality and the matrices $\textbf{U}_1, \dots ,\textbf{U}_m$ having orthonormal columns. The matrices $\textbf{U}_i$ are in HOSVD estimated through singular value decompositions of certain flattenings of the tensor $\mathcal{A}$, in turn allowing the estimation of the core. In practice, the optimal dimensions of the core tensor are usually not known \textit{a priori} and have to be estimated, which is the main objective of the current work. We discuss HOSVD more closely in Section \ref{sec:model} in conjunction with our model of choice. In fact, there we use a specific ``statistical'' version of HOSVD known as $(2\mathrm{D})^2$PCA~\cite{zhang20052d}.

Besides HOSVD, another classical tensor decomposition is the Tucker decomposition \cite{tucker1966some}, which also seeks for an approximation of the previous form but uses a different algorithm for estimating the loading matrices and the core. However, in this paper we have chosen to work with HOSVD as it has a closed-form solution, which is what allows studying the theoretical properties of our proposed estimator of the latent dimensionality.



Finally, we still remark that, arguably, the most typical way of decomposing image data to low-rank components is through standard principal component analysis (PCA) \cite{jolliffe2002principal}. However, as standard PCA operates on vector-valued data and not tensors, this approach requires the vectorization of the data, causing them to lose their natural tensor structure. Whereas, tensor decompositions 
retain the row-column structure of image data, making them the natural choice in the current context.

\subsection{Order determination}

The problem of choosing the (in some sense) optimal rank in tensor decompositions, or in dimension reduction in general, is known as order determination.

The order determination literature can be roughly divided to two categories: (1) methods targeted for specific parametric/semi-parametric models and, (2) more general methods that estimate the rank of a fixed matrix from its noisy estimate (under some regularity conditions).

Category (1) is by far the more popular one (as general methods are more difficult to come by) and usually accomplishes the estimation by exploiting asymptotic properties of eigenvalues of certain matrices, see, e.g., \cite{schott2006high,nordhausen2022asymptotic} for PCA and \cite{bura2011dimension, zhu2006sliced} for dimension reduction in the context of regression (sufficient dimension reduction). Whereas, methods belonging to category (2) tend to be based on various bootstrapping and related procedures, see \cite{YeWeiss2003, LuoLi2016} and, in particular, \cite{LuoLi2021} whose augmentation procedure serves as a starting point for the current work.

We note that all of the previous references targeted exclusively vector-valued data and order determination in the context of general tensor-valued data is indeed still rare in the literature. Moreover, aside from ~\cite{RadojicicLietzenNordhausenVirta2021}, to our best knowledge, automated order determination in the context of matrix-variate (gray-scale image) data has been studied only by \cite{tu2019generalized} who use Stein's unbiased risk estimation for the task. Also, a simpler problem of selecting the dimension when the amount of retained variance is pre-determined is discussed in \cite{hung2012multilinear}. Thus, there is still much work to do in the order determination of tensor data and the current work is a step towards this direction.


\subsection{Relation to previous work}\label{sec:previous}

In this work, we propose an automatic procedure for the selection of the dimensionality in the HOSVD decomposition of tensorial data. Our primary contributions in relation to earlier literature are the following.

\begin{itemize}
    \item[(i)] To our best knowledge, 
    ours is the first order determination procedure for general tensor-valued data with statistical guarantees. 
    \item[(ii)] Unlike \cite{LuoLi2021} who originally introduced the augmentation procedure in vector-valued context, we prove the validity of the procedure also on the population level, see Corollary \ref{cor:cor_1}.
    \item[(iii)] We derive the full limiting distribution for the norms of augmented parts of the noise eigenvectors, see part (ii) of Theorem \ref{thm:consistency_eigenvectors}. This is in strict contrast to \cite[Section 6]{LuoLi2021} who, in an analogous context (PCA) for vector-valued data, based their work on the weaker claim that the norms are non-negligible in probability. 
    This improvement both allows us to get quantitative results concerning the eigenvectors and sheds some light on the interplay of the augmentation procedure with the true dimensionality of the data.
    \item[(iv)] To accompany the augmentation estimator, we also propose an alternative estimator of the latent dimensionality based on the ``ladle'' procedure \cite{LuoLi2016}.
\end{itemize}

Compared to the conference paper \cite{RadojicicLietzenNordhausenVirta2021}, which presented versions of the augmentation and ladle procedures for matrix-valued data and of which the current work is an extended version, we go beyond them in the following respects:

\begin{itemize}
    \item [(i)] We allow for general tensor-valued data, not just matrices.
    \item[(ii)] We give estimation guarantees for the augmentation method both on the population and the sample level (unlike \cite{RadojicicLietzenNordhausenVirta2021} who did not consider the theoretical properties of the method at all).
\end{itemize}





\subsection{Organization of the manuscript}

The rest of the paper is organized as follows. In Section~\ref{sec:model} we introduce the statistical framework along with HOSVD for the proposed model. The proposed augmentation estimator as well as its theoretical properties are discussed in Section~\ref{sec:augmentation}, separately on the population and sample levels. Section~\ref{sec:bootstrap} covers the alternative, bootstrap-based ladle estimator. We conclude the paper by evaluating the performances of the proposed estimators in Section~\ref{sec:numerical_results} and by discussing possible future work in Section \ref{sec:discussion}. A summary of tensor notation is given in Appendix \ref{sec:tensor} and the proofs of all technical results are collected in Appendix~\ref{sec:proofs} in the supplement. 


\section{Model}\label{sec:model}

We use standard tensor notation throughout the manuscript: the calligraphy font $\mathcal{A} $ refers to individual tensors, the $k$-unfolding of a tensor $\mathcal{A}$ is denoted by $\mathcal{A}_k$ and the $k$-mode multiplication of a tensor $\mathcal{A}$ by a matrix $\textbf{A}_k$ is denoted by $\mathcal{A} \times_k \textbf{A}_k$, etc. For readers unfamiliar with tensor notation a summary is given in Appendix \ref{sec:tensor}; see also \cite{Kolda2009}.

We start by defining an appropriate statistical framework. Let $\mathcal{X}^1, \ldots , \mathcal{X}^n \in\mathbb{R}^{p_1\times\dots\times p_m}$ be an observed set of tensors of order $m\in\mathbb{N}$ drawn independently from the model
\begin{align}\label{def:tensor_model}
\mathcal{X}= \mathcal{M} +\mathcal{Z} \times_{k=1}^m \textbf{U}_k +\mathcal{E},
\end{align}
where $\mathcal{M} \in \mathbb{R}^{p_1 \times\dots\times p_m}$ is the mean tensor, 
$\textbf{U}_k\in\mathbb{R}^{p_k \times d_k}$, $k=1,\dots,m$, are unknown mixing matrices with orthonormal columns and $\mathcal{Z}\in \mathbb{R}^{d_1 \times\dots\times d_m}$ is a \textit{core tensor} of order $m$ with zero mean, finite second moment, $\mathbb{E}(\| \mathcal{X} \|_F^2) < \infty$, and dimensions $d_k\leq p_k$, $k=1,\dots,m$. Furthermore, the additive noise $\mathcal{E}$ is assumed to follow a tensor spherical distribution, i.e., $\mathcal{E}\times_{k=1}^m\textbf{V}_k\sim \mathcal{E}$, for all orthogonal matrices $\textbf{V}_k\in\mathbb{R}^{p_k\times p_k}$, $k=1,\dots,m$, where the symbol $\sim$ means ``is equal in distribution to''. Additionally, we make the technical assumption that for all $k$-flattenings $\mathcal{Z}_k\in\mathbb{R}^{d_k\times \rho_k }$ of $\mathcal{Z}$ the matrix $\mathbb{E}(\mathcal{Z}_k\mathcal{Z}_k')\in\mathbb{R}^{d_k\times d_k}$ is positive definite, where $\rho_k := \prod_{j \neq k} p_j$. This assumption is made precisely for the identifiability of the latent dimensions $(d_1, \ldots , d_m)$. A schematic representation of the model is given in Figure~\ref{fig:TensorModel} for $m=3$. A special case of Model~\ref{def:tensor_model} for $m=2$ was studied in~\cite{RadojicicLietzenNordhausenVirta2021}.

\begin{figure}[ht]
    \centering
    \includegraphics[width=0.85\linewidth]{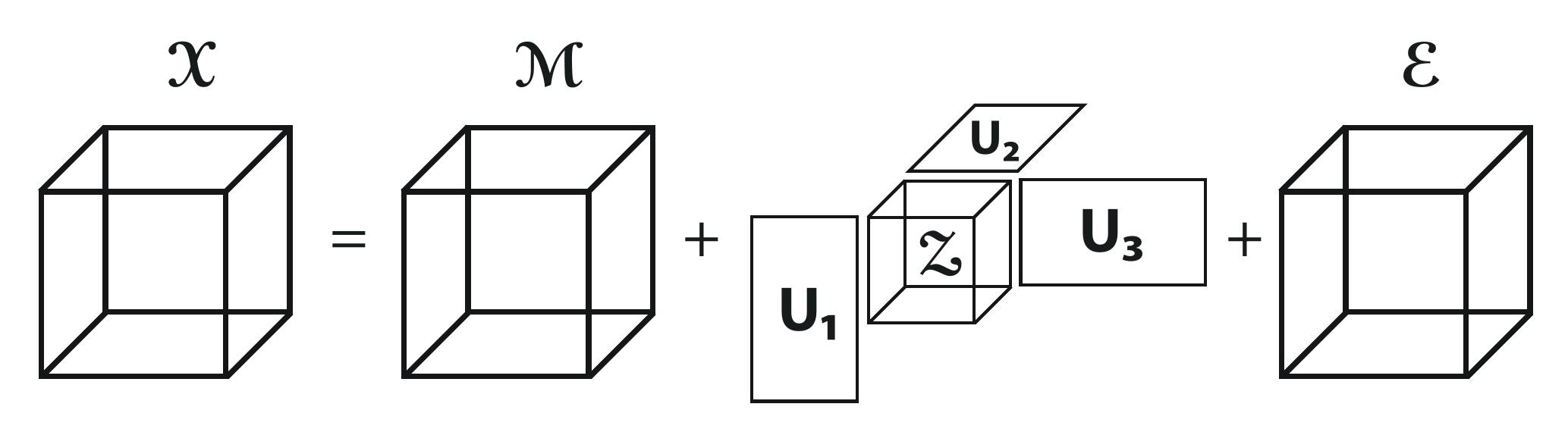}
    \caption{Schematic representation of Model~\ref{def:tensor_model} for $m=3$.}
    \label{fig:TensorModel}
\end{figure}

Our proposed approach is based on a statistical formulation of HOSVD known as $(2\mathrm{D})^2$PCA \cite{zhang20052d}. For this, consider the $k$-flattenings of Model~\ref{def:tensor_model}, for $k=1,\dots, m$,
\begin{align}\label{def:matrix_k_flattening_model}
\mathcal{X}_{k} = \mathcal{M}_k+\textbf{U}_k\mathcal{Z}_k (\textbf{U}^\otimes_{-k})'+\mathcal{E}_k,
\end{align}
where $\mathcal{X}_{k},\, \mathcal{M}_k,\,\mathcal{E}_k\in\mathbb{R}^{p_k\times\rho_k}$, $\mathcal{Z}_k\in\mathbb{R}^{d_k\times\prod_{i\neq k}d_i}$ and $\textbf{U}^\otimes_{-k} := \textbf{U}_{k+1}\otimes\cdots\otimes\textbf{U}_m\otimes\textbf{U}_1\otimes\cdots\otimes\textbf{U}_{k-1}$. Then, the HOSVD solution to Model~\ref{def:tensor_model} is
$(\mathcal{X}-\mathcal{M})\times_{k=1}^m\textbf{V}_k'$, where the columns of $\textbf{V}_k\in\mathbb{R}^{p_k\times d_k}$, $k=1,\dots, m$, are the first $d_k$ eigenvectors of the matrix $\mathbb{E}\{ (\mathcal{X}_k-\mathcal{M}_k)(\mathcal{X}_k-\mathcal{M}_k)'\}$. 
If, for a fixed $k$, the matrix $\mathbb{E}(\mathcal{Z}_k\mathcal{Z}_k')$ is a diagonal matrix with pair-wise distinct diagonal elements, then $\textbf{V}_k$ equals the mixing matrix $\textbf{U}_k$ up to a permutation and sign change of columns. However, even if that is not the case, the column space of the estimated matrix $\textbf{V}_k$ coincides with that of $\textbf{U}_k$.

It is worth mentioning that even though Model~\ref{def:tensor_model} assumes that the core tensor $\mathcal{Z}$ is mixed by matrices $\textbf{U}_1,\dots,\textbf{U}_m$ with \textit{orthonormal} columns, this assumption is without loss of generality since the singular values and the right singular vectors of a general mixing matrix can be absorbed into the core. This, however, has the drawback of making the core tensor identifiable only up to multiplications by invertible matrices from each mode. Nevertheless, the latent dimensionalities $d_k$ are identifiable in this case as well, and therefore we tolerate this ambiguity. Furthermore, this reveals why the proposed method is also a valid pre-processing step for computationally more involved linear feature extraction procedures.

Having estimated (in the previous sense) the parameters $\textbf{U}_k$, we next develop our augmentation estimator for the latent dimensionality.


\section{Augmentation estimator}\label{sec:augmentation}
\subsection{Population-level methodology}

We next describe our basic strategy behind the estimation of the latent dimensions $d_1,\dots,d_m$ using data augmentation. The approach can be seen as a tensorial extension of the method introduced by \cite{LuoLi2021}. Let $\mathcal{X}_k$ be the $k$-flattening \eqref{def:matrix_k_flattening_model}. Then, as $\textbf{U}^\otimes_{-k}$ has orthonormal columns, we have
$$
\mathbb{E}\{(\mathcal{X}_k-\mathcal{M}_k)(\mathcal{X}_k-\mathcal{M}_k)'\}=\textbf{U}_k\mathbb{E}(\mathcal{Z}_k\mathcal{Z}_k')\textbf{U}_k'+\mathbb{E}(\mathcal{E}_k\mathcal{E}_k').
$$

\begin{lemma}\label{lemma:lemma1}
Let $\mathcal{E}\in\mathbb{R}^{p_1\times\cdots\times p_m}$ be a random tensor with tensor spherical distribution. Then, for $k=1,\dots,m$,  $\displaystyle\mathbb{E}(\mathcal{E}_k\mathcal{E}_k')=\sigma_k^2\textbf{I}_{p_k}$, 
for some $\sigma_k^2>0$, where $\mathcal{E}_k$ is the $k$-flattening of $\mathcal{E}$.
\end{lemma}


Using Lemma~\ref{lemma:lemma1} we obtain that
$$
\mathbb{E}\{ (\mathcal{X}_k-\mathcal{M}_k)(\mathcal{X}_k-\mathcal{M}_k)'\} =\textbf{U}_k\mathbb{E}(\mathcal{Z}_k\mathcal{Z}_k')\textbf{U}_k'+\sigma_k^2\textbf{I}_{p_k},
$$
and, since $\mathrm{rank}(\textbf{U}_k\mathbb{E}(\mathcal{Z}_k\mathcal{Z}_k')\textbf{U}_k') = d_k$, the problem of estimating $d_k$ boils down to the problem of estimating the rank of the matrix $\mathbb{E}\{ (\mathcal{X}_k-\mathcal{M}_k)(\mathcal{X}_k-\mathcal{M}_k)'\} - \sigma_k^2\textbf{I}_{p_k}$. A naive way would be to inspect the scree plot of the eigenvalues of the sample estimate of $\mathbb{E}\{ (\mathcal{X}_k-\mathcal{M}_k)(\mathcal{X}_k-\mathcal{M}_k)' \}$ and search for an \textit{elbow}, a point at which the eigenvalues start to even off. However, it is often very difficult and subjective to find such a point. Therefore, we supplement the scree plot with additional information extracted from the eigenvectors of an appropriately constructed matrix, a task in which we employ the augmentation technique demonstrated in~\cite{RadojicicLietzenNordhausenVirta2021}, and initially introduced by~\cite{LuoLi2021} for vectors.

More precisely, for fixed $k = 1, \ldots, m$ and $r_k\in\mathbb{N}$, we define $\textbf{X}_S\in\mathbb{R}^{r_k \times \rho_k}$ to be a random matrix with i.i.d. entries having the distribution $\mathcal{N}(0, \sigma_k^2/\rho_k)$, implying that $\mathbb{E}(\textbf{X}_S)=\textbf{0}$ and $\mathbb{E}(\textbf{X}_S \textbf{X}_S')=\sigma_k^2\textbf{I}_{r_k}$. We then augment (concatenate) the centered $k$-flattening of $\mathcal{X}$ with $\textbf{X}_S$ to obtain the $(p_k + r_k) \times \rho_k$ matrix $\textbf{X}_k^*=((\mathcal{X}_k-\mathcal{M}_k)',\textbf{X}_S')'$ that satisfies,
$$
\mathbb{E}\{\textbf{X}_k^* (\textbf{X}_k^*)'\}
=\begin{pmatrix} \textbf{U}_k\mathbb{E}(\mathcal{Z}_k \mathcal{Z}_k')\textbf{U}_k' &\textbf{0}\\
\textbf{0} & \textbf{0}
\end{pmatrix}+\sigma_k^2\textbf{I}_{p_k+r_k}.
$$
We further define $\textbf{M}_k^* := \mathbb{E}\{ \textbf{X}_k^* (\textbf{X}_k^*)'\}-\sigma_k^2\textbf{I}_{p_k+r_k}$.  Then $\textbf{M}_k^*$ and $\mathbb{E}(\mathcal{Z}_k \mathcal{Z}_k')$ both have the same rank $d_k$ and also the same positive eigenvalues $\lambda_{k,1}\geq \lambda_{k,2}\geq\cdots\geq \lambda_{k,d_k}>0$.

Let next $\boldsymbol\beta_{k,i}^* = (\boldsymbol\beta_{k,i}', \boldsymbol\beta_{k,i,S}')'\in\mathbb{R}^{p_k+r_k}$, $i=1,\dots,p_k+r_k$, be any eigenvector of $\textbf{M}_k^*$ corresponding to its $i$th largest eigenvalue, where we call the $r_k$-dimensional subvector $\boldsymbol\beta_{k,i,S}$ its \textit{augmented part}. Then, for $i\leq d_k$,
$$
\textbf{M}_k^*\boldsymbol\beta_{k,i}^*=(\boldsymbol\beta_{k,i}'\textbf{U}_k\mathbb{E}(\mathcal{Z}_k \mathcal{Z}_k')\textbf{U}_k',\textbf{0}')'=\lambda_{k,i}(\boldsymbol\beta_{k,i}',\boldsymbol\beta_{k,i,S}')',
$$
implying that $\boldsymbol\beta_{k,i,S}=\textbf{0}$ for $i=1,\dots d_k$. Not only does the equivalent fail for the eigenvectors belonging to a zero eigenvalue ($i > d_k$), but the following theorem shows that for $r_k\to\infty$, exactly the opposite happens.

Prior to stating the theorem, let us briefly discuss the form and the arbitrariness of the zero-eigenvalue eigenvectors of $\textbf{M}^*_k$. Let  $\textbf{B}^*_{k,0} \in \mathbb{R}^{(p_k + r_k) \times (p_k + r_k - d_k)}$ denote a matrix that contains an arbitrary orthonormal basis of the null space of $\textbf{M}_k^*$ as its columns (the below result is invariant to the exact choice of this basis). Then, for $i>d_k$, any eigenvector $\boldsymbol{\beta}_{k,i}^*$ lies in the null space of $\textbf{M}_k^*$ and is thus of the form $\boldsymbol{\beta}_{k,i}^*=\textbf{B}_{k,0}^*\textbf{a}$ for some unit length vector $\textbf{a}\in\mathbb{R}^{p_k+r_k-d_k}$. 

The following theorem then shows that the norm of the augmented part of a randomly chosen zero-eigenvalue eigenvector follows a specific beta distribution.

\begin{theorem}\label{thm:thm0}
Fix $i=(d_k + 1),\dots,(p_k+r_k)$ and let $\boldsymbol\beta_{k,i}^* = (\boldsymbol\beta_{k,i}', \boldsymbol\beta_{k,i,S}')'\in\mathbb{R}^{p_k+r_k}$ be of the form $\boldsymbol\beta_{k,i}^* = \textbf{B}^*_{k,0} \textbf{a}$ where $\textbf{a}$ is drawn uniformly from the unit sphere in $\mathbb{R}^{p_k + r_k - d_k}$. Then  $\|\boldsymbol\beta_{k,i,S}\|^2\sim\mathrm{Beta}\{r_k/2,(p_k-d_k)/2\}$.
\end{theorem}


Figure~\ref{fig:fig_beta} illustrates the behaviour of the tail probabilities of the augmented parts of randomly chosen eigenvectors (in the sense of Theorem \ref{thm:thm0})  belonging to a zero eigenvalue, as a function of $r_k$. Note that, in practice, we would like the sample analogues of the quantities $\|\boldsymbol\beta_{k,i,S}\|^2$, $i = d_k + 1, \ldots, p_k$ to be as large as possible to be able to distinguish the transition from signal to noise. Based on Figure~\ref{fig:fig_beta} this can be achieved by using large values of $r_k$. However, this matter actually turns out to be more complicated in a finite-sample case where increasing $r_k$ with $n$ held fixed might lead to high-dimensional phenomena, see Section \ref{sec:discussion}.

\begin{figure*}[ht]
    \centering
    \includegraphics[width=0.95\linewidth]{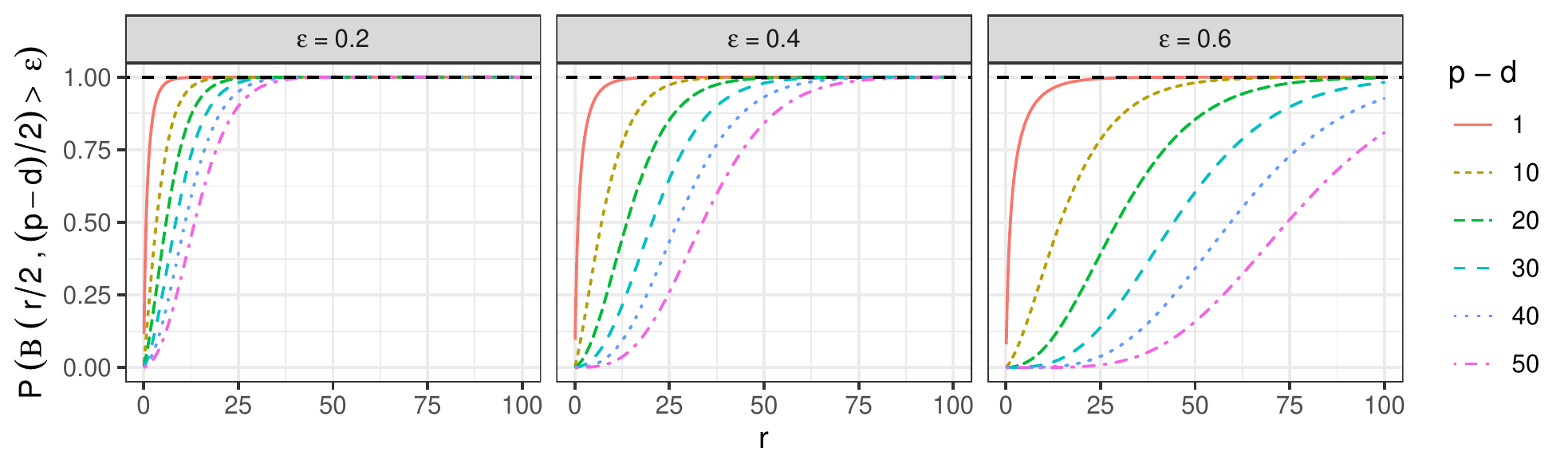}

    \caption{The curves represent the probability that a random variable from $\mathrm{Beta}(r/2,(p-d)/2)$-distribution takes a value larger than $\varepsilon>0$, as a function of the parameter $r$, for various values of $p-d$ and $\varepsilon$. }
    \label{fig:fig_beta}
\end{figure*}

Given the distributional result in Theorem \ref{thm:thm0}, the following properties of the augmented parts of the null eigenvectors of $\textbf{M}_k^*$ now straightforwardly follow.



\begin{cor}\label{cor:cor_1}
Under the conditions of Theorem \ref{thm:thm0}, the following hold.
\begin{itemize}
    \item[(i)] Let $\varepsilon_n$ be any sequence of positive real numbers such that $\varepsilon_n \rightarrow 0$ as $n \rightarrow \infty$. Then, $\mathbb{P}(\|\boldsymbol\beta_{k,i,S}\|^2>\varepsilon_n) \rightarrow 1$ as $n \rightarrow \infty$.
    \item[(ii)]  For fixed $p_k,\,d_k$ and for every $\varepsilon>0$, $\mathbb{P}(\|\boldsymbol\beta_{k, i, S}\|^2\geq 1-\varepsilon)\to 1$, as $r_k\to\infty$.
     \item[(iii)] In a high-dimensional regime, if $p_k-d_k=o(r_k)$,
     then for every $\varepsilon>0$,  $\mathbb{P}(\|\boldsymbol\beta_{k, i, S}\|^2\geq 1-\varepsilon)\to 1$, as $r_k\to\infty$.
\end{itemize}
\end{cor}


Corollary~\ref{cor:cor_1} indicates that for $r_k$ large enough, the norms of the augmented parts of the zero-eigenvalue eigenvectors get arbitrary close to $1$, thus explaining (on the population level) the behavior observed in~\cite{RadojicicLietzenNordhausenVirta2021}, where, when the number of augmentations was increased, the function $\hat{f}_k$ that captures the information from the eigenvectors by accumulating the norms of their augmented parts acted as a linear function of the dimension $j$, for $j>d_1$; see Figure 5 in~\cite{RadojicicLietzenNordhausenVirta2021} for more insight. 

The previous population properties of the augmented parts of the eigenvectors serve as the basis for the construction of the augmentation estimator. Note that the successful estimation of the noise variance $\sigma_k^2$ of the $k$th flattering is crucial for the above construction, and therefore we discuss it next in more detail.

\subsection{Estimation of the noise variance}

Let first $\mathcal{X}^1,\dots,\mathcal{X}^n$ be an i.i.d. sample from Model~\ref{def:tensor_model} and let $\Bar{\mathcal{X}}$ be the corresponding sample mean. 
Furthermore, let $\hat\sigma_{k,1}^2 \geq \cdots \geq \hat\sigma_{k,p_k}^2$ be the eigenvalues of $(1/n) \sum_{i=1}^n (\mathcal{X}_{k,i}-\Bar{\mathcal{X}}_k) (\mathcal{X}_{k,i}-\Bar{\mathcal{X}}_k)'$ and denote by $\sigma_{k,1}^2 \geq \cdots \geq \sigma_{k,p_k}^2$ the eigenvalues of $\mathbb{E}\left\{(\mathcal{X}_k-\mathcal{M}_k)(\mathcal{X}_k-\mathcal{M}_k)'\right\}$. Due to the independence of the signal and the noise, $\sigma_{k,i}^2 = \lambda_{k,i}+\sigma_k^2$ for $i=1,\dots,d_k$, and $\sigma_{k,i}^2 =\sigma_k^2$, for $i=d_k+1,\dots,p_k$, where $\lambda_{k,1},\dots,\lambda_{k,d_k}$ are the eigenvalues of $\mathbb{E}(\mathcal{Z}_k\mathcal{Z}_k')$. This consideration, together with Corollary~\ref{cor:cor_2} in Section \ref{sec:asymp}
implies how we can justly use $\hat\sigma^2_{k,d_k+1},\dots,\hat\sigma^2_{k,p_k}$ to construct a consistent estimator of the noise variance $\sigma_k^2$ of the $k$th mode. However, since it is mostly the case that one wishes to estimate the latent dimensions in all modes, the following discussion allows us to construct a pooled estimator of the noise variance using all modes.\\


The matrix $\mathcal{E}_k$ is left spherical for each $k=1,\dots m$, thus making
$\mathbb{E}(\mathcal{E}_k\mathcal{E}_k')=\sigma_k^2\textbf{I}_{p_k}$,  $\sigma_k^2>0$. Therefore,
\begin{equation}\label{eq:diag_e_k}
\sigma_k^2=\mathbb{E}(\mathcal{E}_k\mathcal{E}_k')_{i,i}=\sum_{j=1}^{\rho_k
}\mathbb{E}(\mathcal{E}_{k,(i,j)}^2).
\end{equation}
If we sum the identity~\eqref{eq:diag_e_k} over all $i=1,\dots,p_k$, we obtain
$ p_k\sigma_k^2=\sum_{i=1}^{p_k}\sum_{j=1}^{\rho_k
}\mathbb{E}(\mathcal{E}_{k,(i,j)}^2)=\mathbb{E}\|\mathcal{E}\|_\mathrm{F}^2,$ thus implying the relationship
\begin{equation}\label{eq:sigmas_relation}
p_1\sigma_1^2=p_2\sigma_2^2=\cdots=p_m\sigma_m^2,
\end{equation}
between the noise variances of the $k$-flattenings $\mathcal{E}_k$ of the noise tensor $\mathcal{E}$.

Define now $S_k :=\{\frac{p_i}{p_k}\sigma_{i,j}^2:i=1,\dots,m,\,j=1,\dots,p_i\}$ to be the set of eigenvalues from all modes in the ``scale'' of the $k$th mode. Similarly,  define $\hat S_k :=\{\frac{p_i}{p_k}\hat\sigma_{i,j}^2:i=1,\dots,m,\,j=1,\dots,p_i\}$, to be the sample counterpart of $S_k$. 
Lemma~\ref{lemma:lemma_consistency_of_estimators} in Section \ref{sec:asymp} shows that under certain assumptions on the compressibility of the data, quantiles as well as means of suitable tails of $\hat{S}_k$ are consistent estimators of $\sigma_k^2$.
Naturally, once the noise variance of the $k$th mode has been estimated, we obtain estimates $\hat{\sigma}_i^2$, $i\neq k$, for the noise variances of the remaining modes simply by scaling $\hat{\sigma}_i^2=(p_k/p_i)\hat{\sigma}_k^2$, $i\neq k$. 

\begin{remark}\label{rem:rem1}
To further clarify the scaling constants $p_i/p_k$, $i=1,\dots,m$, used in the estimation of the noise variance in the $k$th mode, consider a scenario where the entries of $\mathcal{E}$ are uncorrelated with zero mean and variance $\delta^2>0$. Then, for $k=1,\dots,m$,   $\mathrm{E}(\mathcal{E}_k\mathcal{E}_k')=\sum_{i=1}^{\rho_k
}\delta^2\textbf{I}_{p_k}=\rho_k
\delta^2\textbf{I}_{p_k}$, thus showing that the noise variance accumulates with the number of columns.
\end{remark}

\subsection{Sample-level estimation}\label{subsec:augmentation_ladle}

We are now equipped to define the augmentation estimator for estimation of the $k$th latent dimension $d_k$. Let $\textbf{X}_{1,S},\dots,\textbf{X}_{n,S}$ be a sample of i.i.d. $r_k\times \rho_k
$ matrices with elements drawn from the standard normal distribution $\mathcal{N}(0,1)$. We define the augmented $k$-flattenings of the observations $\mathcal{X}^1,\dots,\mathcal{X}^n$ as the $ (p_k+r_k)\times \rho_k$ matrices $\textbf{X}_{i,k}^* := ((\mathcal{X}_{k}^i)',\hat{\sigma}_k\textbf{X}_{i,S}')'$, $i=1,\dots,n$, where $\hat{\sigma}_k^2$ is any consistent estimator of the noise variance $\sigma_k^2$, see Lemma \ref{lemma:lemma_consistency_of_estimators} in Section \ref{sec:asymp} for examples. Let further $\bar{\textbf{X}}_k^*$ be the sample mean of the obtained augmented sample. A sample estimate $\hat{\textbf{M}}_k^*$ of the matrix $\textbf{M}_k^*$ is then
$$
\hat{\textbf{M}}_k^*=\frac{1}{n}\sum_{i=1}^n(\textbf{X}_{i,k}^*-\bar{\textbf{X}}_k^*)(\textbf{X}_{i,k}^*-\bar{\textbf{X}}_k^*)'-\hat\sigma_k^2\textbf{I}_{p_k+r_k},
$$
whose first $p_k$ eigenvectors we denote in the following by $\hat{\boldsymbol\beta}_{k,1}^*, \dots, \hat{\boldsymbol\beta}_{k,p_k}^*$. Mimicking \cite{LuoLi2021} and \cite{RadojicicLietzenNordhausenVirta2021}, we define the normalized scree plot curve,
$$
\hat{\Phi}_k:\{0,1,\dots,p_k\}\to \mathbb{R},\quad \hat{\Phi}_k(l)=\hat\lambda_{k,l+1}/\left(\sum_{i=1}^{l+1}\hat\lambda_{k,i}+1\right),
$$
where $(\hat{\lambda}_{k,1}, \ldots , \hat{\lambda}_{k,p_k}) := (\hat{\sigma}_{k,1}^2 - \hat{\sigma}_k^2, \ldots , \hat{\sigma}^2_{k,p_k} - \hat{\sigma}_k^2)$, 
and we take $\hat\lambda_{k,p_k+1} := 0$. However, as the values $\hat{\sigma}_{k,i}^2 - \hat{\sigma}_k^2$ are not necessarily non-negative (unlike their population counterparts, the eigenvalues of $\mathbb{E}(\mathcal{Z}_k\mathcal{Z}_k')$), we suggest using $\hat\lambda_{k,i}=\max\{\hat\sigma_{k,i}^2-\hat\sigma_k^2,0\}$, $i = 1, \ldots , p_k$, instead. In any case, one should proceed with caution as very negative values of $\hat\sigma_{k,i}^2-\hat\sigma_{k}^2$ indicate possible overestimation of the noise variance $\sigma_k^2$. Lemma~\ref{lemma:lemma_consistency_of_estimators} in Section \ref{sec:asymp} lists a number of consistent estimators of noise variance which, however,  possibly behave rather differently. Thus, in Remark~\ref{rem:sigma_missestimation} we discuss the effect of misestimation of the noise variance for the presented procedure.

\begin{remark}\label{rem:sigma_missestimation}
Let $\textbf{X}_S\in\mathbb{R}^{r_k\times\rho_k}$ be the augmented submatrix for the $k$th flattening $\mathcal{X}_k$, having independent $\mathcal{N}(0,\sigma_S^2/\rho_k)$-elements, where $\sigma^2_S > 0$ is now understood to be the (fixed) estimated value of $\sigma_k^2$. 
Furthermore, \begin{align*}
\textbf{M}_k^*&=\mathbb{E}\{(\textbf{X}_k^*-\mathbb{E}(\textbf{X}_k^*)) (\textbf{X}_k^*-\mathbb{E}(\textbf{X}_k^*))'\}-\sigma_S^2\textbf{I}_{p_k+r_k}\\
&=\begin{pmatrix} \textbf{U}_k \{ \mathrm{E}(\mathcal{Z}_k \mathcal{Z}_k')+(\sigma_k^2-\sigma_S^2)\textbf{I}_{p_k} \} \textbf{U}_k' &\textbf{0}\\
\textbf{0} & \textbf{0}
\end{pmatrix},
\end{align*}
and the eigenvalues of $\textbf{M}_k^*$ are $\lambda_{k,i}+(\sigma_k^2-\sigma_S^2)$, $i=1,\dots,p_k$, where $\lambda_{k,i}=0$ for $i>d_k$, in addition to the $r_k$ zero eigenvalues corresponding to the lower right block. In practice, since the eigenvalues of $\textbf{M}_k^*$ serve as the estimators of the eigenvalues of the positive-definite matrix $\mathbb{E}(\mathcal{Z}_k\mathcal{Z}_k')$, as discussed in Section~\ref{subsec:augmentation_ladle}, we replace $\lambda_{k,i}+(\sigma_k^2-\sigma_S^2)$ with  $\max\{0,\lambda_{k,i}+(\sigma_k^2-\sigma_S^2)\}$, $i=1,\dots,p_k$,  to 
avoid negative values. Let now $\sigma_S^2=\sigma_k^2+\delta$, where $0\leq\delta<\lambda_{k,d_k}$ and $\delta>0$ corresponds to the amount of overestimation of $\sigma_k^2$. Then, $\max\{0,\lambda_{k,i}+(\sigma_k^2-\sigma_S^2)\}=\lambda_{k,i}-\delta$, $i=1,\dots,d_k$, and $\max\{0,\lambda_{k,i}+(\sigma_k^2-\sigma_S^2)\}=0$, for $i>d_k$, implying that the thresholding preserves the rank $d_k$ while shifting the nontrivial eigenvalues by $-\delta$.  
\end{remark}

Thus, Remark \ref{rem:sigma_missestimation} shows that the method is robust towards slight overestimation of the noise variance, where such behavior is directly related to the thresholding of the eigenvalues of $\hat{\textbf M}_k^*$ from below by $0$. The ``allowed'' amount of overestimation is equal to the smallest non-trivial eigenvalue of $\mathbb{E}(\mathcal{Z}_k\mathcal{Z}_k')$. Though Remark~\ref{rem:sigma_missestimation} explains the effect of the overestimation of the noise variance at the population level, approximation to the phenomenon holds in the sample case.

Moving back to the estimation of the order $d_k$, additional information about it can now be obtained by using the eigenvectors of $\textbf{M}_k^*$. To reduce the effect of randomness in the augmentation, the augmentation procedure is conducted independently
$s_k$ times, and we compute the eigenvectors of $\hat{\textbf{M}}_k^{*}$ for each replicate. For $j=1,\dots,s_k$, we denote by $\hat{\boldsymbol\beta}_{k,i,S}^j$ the augmented part
of the $i$th eigenvector of the matrix $\hat{\textbf{M}}_k^{*j}$ in the $j$th replicate. The eigenvector information is then captured by the function
$$
\hat{f}_k:\{0,1,\dots,p_k\}\to \mathbb{R},\quad \hat{f}_k(i)=\frac{1}{s_k}\sum_{j=1}^{s_k}\|\hat{\boldsymbol\beta}_{k,i,S}^{j}\|^2,
$$
where $\hat{\boldsymbol\beta}_{k,0,S}^{j}:=\textbf{0}$ . 
Finally, we combine the eigenvalue information captured by $\hat{\Phi}_k$ and the eigenvector information in $\hat{f}_k$ to form the final objective function $\hat{g}_k:\{0,1,\dots,p_k\}\to \mathbb{R}$,
\begin{equation}\label{eq:phi_aug}
    \hat{g}_k(j)=\hat{\Phi}_k(j)+\sum_{i=0}^j \hat{f}_k(i) ,
\end{equation}
whose minimizer $\hat d_k$ is taken to be the estimator of the latent dimension $d_k$. This definition of $\hat d_k$ is intuitively clear as, assuming that $d_k > 0$, for any $i<d_k$ the eigenvalue part $\hat{\Phi}_k(i)$ of~\eqref{eq:phi_aug} is large, while the eigenvector part $\hat{f}_k(i)$ is small. For $i > d_k$, the opposite happens and the eigenvalue part is small while the eigenvector part is large, due to Corollary~\ref{cor:cor_3} in Section \ref{sec:asymp}. At the true dimension $i=d_k$ both parts are small, thus implying that the sum curve $\hat{g}_k$ is (at the population level) minimized precisely at $i = d_k$. In the extreme noise case where $d_k=0$, the eigenvalue part in~\eqref{eq:phi_aug} is always negligible, while again due to Corollary~\ref{cor:cor_3}, the eigenvector part is always large, except for $i=0$, in which case it vanishes, causing the minimum to occur at $i=0$. An algorithm for the augmentation estimator is given in Algorithm~\ref{alg::aug} and the augmentation process is visualized in Figure~\ref{fig:AugPorc}.

\begin{algorithm}[ht]
\caption{Augmentation estimator for the dimension $d_k$ of the $k$th mode}\label{alg::aug}
\SetKwInOut{Input}{Input}
        \Input{$\mathcal{X}^1,\dots, \mathcal{X}^n\in\mathbb{R}^{p_1\times\cdots\times p_m}$, centered sample of tensors;}
        \BlankLine
	Set the row dimension $r_k>0$;\\
 	Set the number of augmented replicates $s_k>0$;\\

    Calculate $\hat{\textbf{M}}_k=\frac{1}{n}\sum_{i=1}^n\mathcal{X}_{k}^i{\mathcal{X}_{k}^i}'$, for $k=1,\dots,m$;\\

	Calculate the estimate $\hat\sigma_k^2$ of the noise variance based on the pooled set of scaled eigenvalues of $\hat{\textbf{M}}_k$, $\displaystyle
	\hat{S}_k=\{\frac{p_i}{p_k}\hat{\sigma}_{i,j_i}^2:i=1,\dots,m,\,j_i=1,\dots,p_i\}.
	$

    Compute $\hat\lambda_{k,i} = \max\{\hat\sigma_{k,i}^2-\hat\sigma_k^2,0\}$;

	\For{$i\gets 1$ \KwTo $n$}{
		\For{$j\gets 1$ \KwTo $s_k$}{
        Generate an $r_k \times \rho_k$ matrix $\textbf{X}_{i,S}^j$, with entries drawn i.i.d. from $\mathcal{N}(0, 1)$;

        Define the augmented $i$th observation as
            $\textbf{X}_i^{*j}={({\mathcal{X}_{k}^i}', \hat\sigma_k {\textbf{X}_{i,S}^j}')}';$
            }
    }
    \For{$j\gets 1$ \KwTo $s_k$}{
        Compute the eigendecomposition of the $j$th replicated matrix
        \[
        \hat{\textbf{M}}_k^{*j}=\frac{1}{n}\sum_{i=1}^n\textbf{X}_i^{*j}{\textbf{X}_i^{*j}}'-\hat\sigma_k^2\textbf{I}_{p_k + r_k}.
        \]

        Let $\hat{\boldsymbol\beta}_{k,i,S}^{j}$ be the augmented part of the $i$th eigenvector of $ \hat{\textbf{M}}_k^{*j}$;
    }

    Compute the objective function $\displaystyle \hat{g}_k(j)=\hat{\Phi}_k(j)+\sum_{i=0}^j \hat{f}_k(i),$
    where $\hat{\boldsymbol\beta}_{k,0,S}^{j}=\textbf{0}$ and $\hat\lambda_{k,p_k+1}=0$;\\
	
	Return $\hat{d}_k=\mathrm{argmin}\{\hat{g}_k(i):\,i=0,\dots,p_k\}$;
\end{algorithm}

\begin{figure*}
    \centering
    \includegraphics[width=0.85\linewidth]{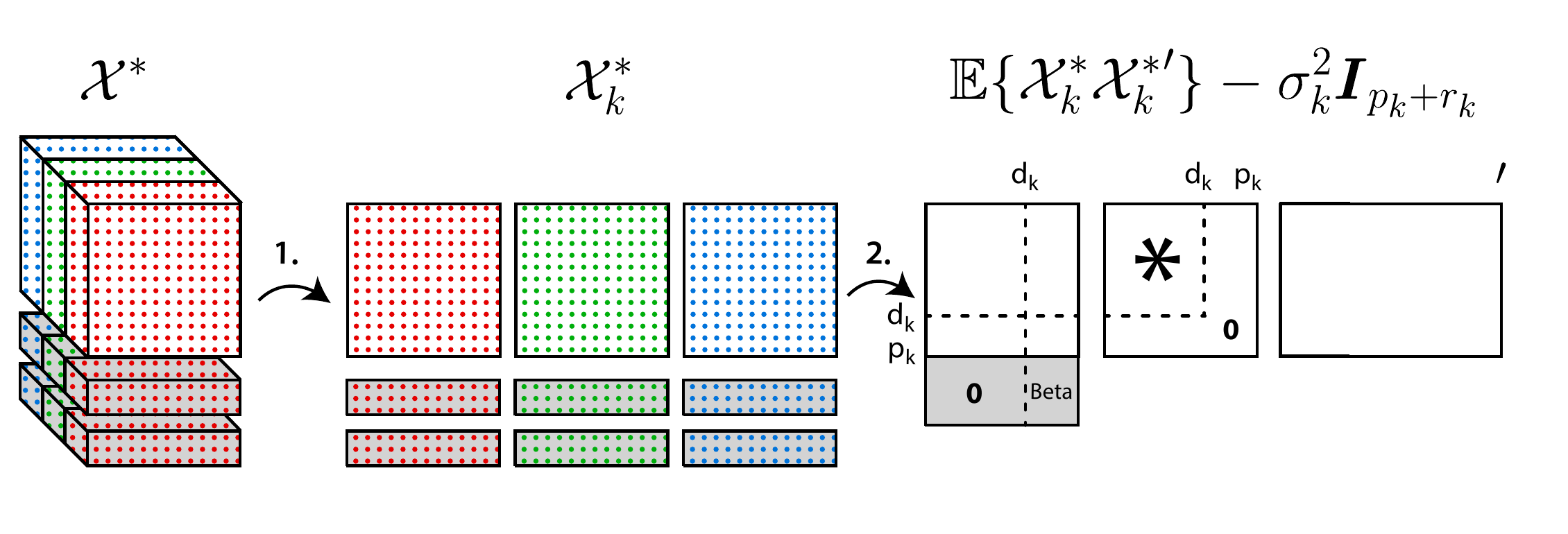}
    \caption{Representation of the augmentation process in the case $m=3$. The parts of the tensors and matrices corresponding to the augmentation have grey background. Step 1 represents the flattening along mode $k$ and Step 2 the subsequent computation of the scatter in that mode along with thr corresponding eigenvalue-eigenvector decomposition.}
    \label{fig:AugPorc}
\end{figure*}

\subsection{Asymptotic properties of the augmentation estimator}\label{sec:asymp}


The following corollary gives the asymptotic behavior of the eigenvalue part $\hat{\Phi}_k$ of $\hat{g}_k$, justifying its use.

\begin{cor}\label{cor:cor_2}
Let $\lambda_{k,i}$ and $\hat\lambda_{k,i}$, $i=1,\dots,p_k+r_k$ be the eigenvalues of $\textbf{M}_k^*$ and $\hat{\textbf{M}}_k^*$, respectively. Then, $\hat\lambda_{k,i}\to_{P}\lambda_{k,i}$ for $i=1,\dots,p_k+r_k$, as $n \rightarrow \infty$.
\end{cor}


An implication of Corollary~\ref{cor:cor_2} is Lemma~\ref{lemma:lemma_consistency_of_estimators}, that lists a number of consistent estimators of the noise variance.

\begin{lemma}\label{lemma:lemma_consistency_of_estimators}
Let $\hat\sigma_{k,q}^2$, $q\in(0,1)$, be the $q$th sample quantile of $\hat S_k$ and $\bar\sigma_{k,q}^2$, $q\in(0,1)$, be the mean of those elements of $\hat S_k$ that are smaller than or equal to $\hat\sigma_{k,q}^2$.
\begin{itemize}
    \item[i)] If $d_1+\dots+d_m<(1-q)(p_1+\dots+p_m)$, then $\hat\sigma^2_{k,q_1}$ and $\bar\sigma_{k,q_1}^2$, for any $q_1\leq q$, are consistent estimators of $\sigma^2_k$. 
    \item[ii)] If $d_1+\dots+d_m<p_1+\dots+p_m$, then $\min\{\hat S_k \}$ is a consistent estimator of $\sigma^2_k$.
\end{itemize}
\end{lemma}

The following theorem illustrates the behaviour of the norms of the augmented parts of the eigenvectors on the sample level, under the assumption of normality for the additive noise $\mathcal{E}$ in Model~\ref{def:tensor_model}, and shows (i) that for $i\leq d_k$ the norms are negligible in probability and (ii) that this is not the case for the later eigenvectors.

\begin{assumption}\label{ass:Gaussian_noise}
The additive noise $\mathcal{E}$ in Model~\eqref{def:tensor_model} has i.i.d. Gaussian entries.
\end{assumption}

\begin{theorem}\label{thm:consistency_eigenvectors}
Let $\hat{\boldsymbol\beta}_{k,i}^*=(\hat{\boldsymbol\beta}_{k,i,1}',\hat{\boldsymbol\beta}_{k,i,S}')'$, $i=1,\dots, p_k+r_k$ be any set of eigenvectors of $\hat{\textbf{M}}_k^*$,
where $\hat{\boldsymbol\beta}_{k,i,S}\in\mathbb{R}^{r_k}$ is the augmented part of the $i$th eigenvector of $\hat{\textbf{M}}_k^*$. Then,
\begin{itemize}
    \item[(i)] $\|\hat{\boldsymbol\beta}_{k,i,S}\|^2=o_P(1)$, $i\leq d_k$.
    \item[(ii)] If additionally Assumption~\eqref{ass:Gaussian_noise} is satisfied, then $ \|\hat{\boldsymbol\beta}_{k,i,S}\|^2 \rightsquigarrow \mathrm{Beta}\{ r_k/2, (p_k - d_k)/2 \}$ for $i > d_k$.
\end{itemize}
\end{theorem}

Corollary~\ref{cor:cor_3} further illustrates the behaviour of the norms of the augmented parts $\hat{\boldsymbol{\beta}}_{k,i,S}$ for
$i > d_k$. 

\begin{cor}\label{cor:cor_3}
Let $\hat{\textbf{M}}_k^*$ be as defined above and let $\hat{\boldsymbol\beta}_{k,i}^* = (\hat{\boldsymbol\beta}_{k,i,1}', \hat{\boldsymbol\beta}_{k,i,S}')'\in\mathbb{R}^{p_k+r_k}$, $i=1,\dots,p_k+r_k$, be an eigenvector of $\hat{\textbf{M}}_k^*$ corresponding to its $i$th eigenvalue. Then, under  Assumption~\eqref{ass:Gaussian_noise}, for $i>d_k$ and for every $\varepsilon>0$, $\displaystyle\lim_{\varepsilon\to 0^+}\mathbb{P}(\|\hat{\boldsymbol\beta}_{k,i,S}\|^2>\varepsilon)\to 1$, as $n\to\infty$.
\end{cor}

Interestingly, the limiting distribution of $\|\hat{\boldsymbol{\beta}}_{k,i,S}\|^2$, $i>d_k$, in Theorem~\ref{thm:consistency_eigenvectors} does not depend directly on the dimension $p_k$, but rather on the ``amount of noise'' $p_k-d_k$ in the $k$th mode. This implies that the more noise components there are in the $k$th mode, the more difficult it is to differentiate between the signal and the noise eigenvectors using the norms of the corresponding augmented parts; see Figure~\ref{fig:fig_beta} for more insight.

Finally, the following theorem proves the validity of the method, in the sense of the consistency of the estimated latent dimensions, under the assumption of normality of the additive noise.

\begin{theorem}\label{thm:consistency_of_d}
Let Assumption~\eqref{ass:Gaussian_noise} be satisfied and let $\hat d_k$ be the estimator of the unknown dimension $d_k$ defined in~\eqref{eq:phi_aug}. Then,
$$
\lim_{n\to\infty}\mathbb{P}(\hat d_k=d_k)=1, \quad k=1,\dots,m.
$$
\end{theorem}

\section{Bootstrap-based ladle estimator}\label{sec:bootstrap}

As a competitor to the augmentation strategy presented in Section~\ref{subsec:augmentation_ladle}, we introduce a generalization of the bootstrap-based ``ladle''-technique for extracting information from the eigenvectors of the variation matrix presented in~\cite{LuoLi2016} for vector-valued observations. The general idea is to use bootstrap resampling techniques to approximate the variation of the span of the first $k$ eigenvectors of the corresponding sample scatter matrix, where high variation of the span indicates that the chosen eigenvectors belong to the same eigenspace, i.e., that the difference between the corresponding eigenvalues is small, see~\cite{YeWeiss2003}. The information obtained from the eigenvectors is then combined with the one from the eigenvalues of the variation matrix, as in the augmentation estimator. As in \cite{LuoLi2016}, we refer to this composite estimator as the bootstrap ladle estimator, due to a specific ladle shape of the associated plots. 
Namely, a ladle-shaped curve is obtained when the bootstrap ladle estimator is plotted as a function of the unknown dimension. Extending the bootstrap ladle estimator of \cite{LuoLi2016} to tensor-valued observations, we obtain estimators for the orders $d_k$, $k=1,\dots,m$, where the information contained in the eigenvalues is extracted similarly as in Algorithm~\ref{alg::aug}.

More precisely, for centered, independent realizations $\{\mathcal{X}^1,\dots,\mathcal{X}^n\}$ of a zero-mean tensor from Model~\eqref{def:tensor_model}, let $\hat{\textbf{M}}_k$ and $\hat{\sigma}_{k,i}^2$, $k=1,\dots,m$, $i=1,\dots,p_k$, be defined as in Algorithm~\ref{alg::aug}, and let $\hat{\textbf{B}}_{j,k}$ be a matrix containing any first $j$ eigenvectors of $\hat{\textbf{M}}_k$. We define further $\hat{\phi}_{k,\rm{B}}:\{0,1,\dots,p_k-1\}\to \mathbb{R},$ with
\begin{align*}
 \hat{\phi}_{k,\rm{B}}(j)=\hat\sigma_{k,j+1}^2/\left(\sum_{i=1}^{p_k-1}\hat\sigma_{k,i}^2+1\right),
\end{align*}
observing that the value of $\hat{\phi}_{k,\rm{B}}(j)$ is large for $j<d_k$ and small, but not zero, for $j\geq d_k$, for the reasons presented in Section~\ref{subsec:augmentation_ladle}.

As mentioned earlier, the eigenvalue information is in ladle supplemented with information taken from the eigenvectors of $\hat{\textbf{M}}_k$ using a bootstrap technique. For the $i$-th centered bootstrap sample, 
$\{\mathcal{X}_{i}^{1*},\dots,\mathcal{X}_{i}^{n*}\}$, let $\textbf{M}_{k}^{i*}=\frac{1}{n}\sum_{j=1}^n\mathcal{X}_{i,k}^{j*}(\mathcal{X}_{i,k}^{j*})'$ be the scatter matrix of the $k$-flattening of the $i$th bootstrap sample. 
Furthermore, let   $\textbf{B}_{j,k}^{i*}=(\boldsymbol\beta_{1,k}^{i*},\dots,\boldsymbol\beta_{j,k}^{i*})$ be a matrix of any first $j$ eigenvectors, belonging to the $j$ largest eigenvalues of $\textbf{M}_k^{i*}$, $j=1,\dots , p_k-1$. Define further $\hat{f}_{k,\rm{B}}:\{0,1,\dots,p_k-1\}\to \mathbb{R}$, with $\hat{f}_{k,\rm{B}}(0):=0$ and
\[
\hat{f}_{k,\rm{B}}(j)=\frac{1}{s_k}\sum_{i=1}^{s_k}\left(1-|\det(\hat{\textbf{B}}_{j,k}'\textbf{B}_{j,k}^{i*})|\right),\,\text{ for }j>0,
\]
where $s_k$ is the number of bootstrap samples. Note that $\hat{f}_{k,\rm{B}}(p_k)=0$, since $\hat{\textbf{B}}_{j,k}$ and  $\textbf{B}_{j,k}^{i*}$ both span the same space $\mathbb{R}^{p_k}$. The bootstrap ladle estimator $\hat d_k$ of $d_k$ is then defined as the minimizer of $\hat{g}_{k,\rm{B}}:\{0,1,\dots,p_k-1\}\to\mathbb{R}$, where 
\[
\hat{g}_{k,\rm{B}}(j)=\hat{\phi}_{k,\rm{B}}(j)+\hat{f}_{k,\rm{B}}(j)/\left(\sum_{i=1}^{p_k-1}\hat{f}_{k,\rm{B}}(i)+1\right ),
\]
see Algorithm \ref{alg::boot} for a detailed algorithm.

As the value of the objective function $\hat{g}_{k,\rm{B}}$ is not defined at $p_k$, it is assumed that compression is possible in every mode, which is different from the augmentation estimator where compression in at least one mode is compulsory.
\cite{LuoLi2016} argue further that if $p_k$ is large (e.g., $p_k>10$), the normalizing constants for the eigenvalue and eigenvector parts of the objective function should be replaced by $\sum_{i=1}^{\floor{p_k/\log(p_k)}}\hat{\sigma}^2_{k,i}+1$ and $\sum_{i=1}^{\floor{p_k/\log(p_k)}}\hat{f}_{k,\rm{B}}(i)+1$, respectively. The reason is that if $p_k$ is very large, the normalization constant of the eigenvector part of $\hat{g}_{k,\rm{B}}$ increases and weights down the bootstrap part compared to $\hat{\phi}_{k,\rm{B}}$. Therefore, it is usually beneficial to optimize $\hat{g}_{k,\rm{B}}$ only up to some $q_k<p_k$, and $q_k=\floor{p_k/\log(p_k)}$ seems, according to \cite{LuoLi2016}, justifiable in many applications.

\begin{algorithm}[ht]
\caption{Bootstrap ladle estimator for the dimension $d_k$ of the $k$th mode.}\label{alg::boot}
\SetKwInOut{Input}{Input}
        \Input{$\mathcal{X}^1,\dots, \mathcal{X}^n\in\mathbb{R}^{p_1\times\cdots\times p_m}$, centered sample of tensors;}
        \BlankLine
	Set the number of bootstrap samples $s_k>0$;\\
    Calculate the scatter $\hat{\textbf{M}}_{k}$ of the $k$-flattening of the sample $\mathcal{X}^1,\dots, \mathcal{X}^n$;\\
    Calculate the eigendecomposition of $\hat{\textbf{M}}_{k}$ and denote by  $\hat{\textbf{B}}_{j,k}$ any matrix of first $j$ eigenvectors of $\hat{\textbf{M}}_k$;

	\For{$i\gets 1$ \KwTo $s_k$}{
       Sample with repetition the $i$-th (centered) bootstrap sample $\{\mathcal{X}_{i}^{1*},\dots,\mathcal{X}_{i}^{n*}\}$;\\
       Calculate the scatter $\textbf{M}_{k}^{i*}=\frac{1}{n}\sum_{j=1}^n\mathcal{X}_{i,k}^{j*}(\mathcal{X}_{i,k}^{j*})'$ of the $k$-flattening of the $i$th bootstrap sample;\\

       Calculate the eigendecomposition of $\textbf{M}_{k}^{i*}$ and take $\textbf{B}_{j,k}^{i*}$ to be any matrix of first $j$ eigenvectors of $\textbf{M}_k^{i*}$;
    }
    \uIf{$p_k\leq 10$}{Set $q_k\leftarrow p_k-1$\; }
  \Else{Set $q_k\leftarrow\floor{p_k/\log(p_k)}$\;}

    The objective function is $ \hat{g}_{k,\rm{B}}(j):\{0,1,\dots q_k\}\to\mathbb{R}$, with $\hat{g}_{k,\rm{B}}(0)=0$, and, for $j>0$, $\displaystyle
   \hat{g}_{k,\rm{B}}(j)=$
   $\hat\sigma_{k,j+1}^2/\left(\sum_{i=1}^{q_k}\hat\sigma_{k,i}^2+1\right)+\hat{f}_{k,\rm{B}}(j)/\left(\sum_{i=1}^{q_k}\hat{f}_{k,\rm{B}}(i)+1\right );$\\
	
	Return $\hat{d}_k=\mathrm{argmin}\{\hat{g}_{k,\rm{B}}(j):\,i=0,\dots,q_k\}$;

\end{algorithm}

A clear advantage of the bootstrap ladle estimator over the augmentation one is that no noise variance estimation is required and one can use the matrix $\hat{\textbf{M}}_k$ directly to draw inference on the order $d_k$. However, as shown in a simulation study in~\cite{RadojicicLietzenNordhausenVirta2021}, it is computationally more demanding and less accurate and is therefore omitted from the real data analysis in Section~\ref{sec:real_data_example} (but it is still included in the simulations in Section \ref{sec:Simulations}).

Further detailed intuition on the bootstrap ladle estimator can be found in~\cite{LuoLi2016}.

\section{Numerical results}\label{sec:numerical_results}
The analysis of both simulated and real data  was conducted using R~\cite{R} jointly with the packages ICtest~\cite{RICtest}, MixMatrix~\cite{RMixMode} and tensorBSS~\cite{RtensorBSS}. To obtain the augmentation and ladle estimates of the latent dimensions we use
Algorithms~\ref{alg::aug}
and~\ref{alg::boot}, respectively, which are available in the package tensorBSS~\cite{RtensorBSS}. 
\subsection{Simulation study}\label{sec:Simulations}
In the simulation study, the data are generated from the  model
\begin{equation}\label{eq:model_simus}
\mathcal{X}=\mathcal{Z}\times_{i=1}^3\textbf{U}_i+\mathcal{E},
\end{equation}
where $\mathcal{Z}=\mathcal{Z}_0\times_{i=1}^3\textbf{A}_i$, $\mathcal{E}=\sigma\mathcal{E}_0\times_{i=1}^3\textbf{V}_i$ and $\mathcal{Z}_0$ has i.i.d. $t(3)$ entries, $\mathcal{E}_0$ has i.i.d. $\mathcal{N}(0,1)$ entries, $\textbf{A}_i\in\mathbb{R}^{d_i\times d_i}$ 
and $\textbf{V}_i\in\mathbb{R}^{p_i\times p_i}$, $i=1,2,3$, are full column-rank and orthogonal matrices, respectively, with $d_1 = 3,\, d_2 = 5,\, d_3 = 10,\, p_1 = 5,\, p_2 = 15,\, p_3 = 20$. Furthermore, we consider three different values for the noise variance $\sigma^2=0.1,\,0.5,\,1.$
Thus,
\begin{equation}
\mathbb{E}(\mathcal{Z}_k\mathcal{Z}_k')=
\tfrac{5}{3}\textbf{A}_k\textbf{A}_k'\prod_{i\neq k}\mathrm{tr}(\textbf{A}_i\textbf{A}_i'),\,\mathbb{E}(\mathcal{E}_k\mathcal{E}_k')=\sigma^2\prod_{i\neq k}p_i\textbf{I}_{p_k}.
\end{equation}
The eigenvalues of $\mathbb{E}(\mathcal{Z}_k\mathcal{Z}_k')$ are for mode 1 $\{$5.75, 12.93, 22.99$\}$, for mode 2 $\{$5.39,  5.94,  8.41,  9.81, 12.12$\}$ and for mode 3 $\{$2.74, 3.02, 3.31, 3.62, 3.94, 4.28, 4.63, 4.99, 5.37, 5.76$\}$. The eigenvalues of $\mathbb{E}(\mathcal{E}_k\mathcal{E}_k')$, for $k=1,\,2,\,3$ and $\sigma^2\in\{0.1,0.5,1\}$, along with the resulting
signal-to-noise ratios (SNR) for Model~\eqref{eq:model_simus} are given in Table~\ref{tab:tab23} 
where we define $\mathrm{SNR}_k=\|\mathbb{E}(\mathcal{Z}_k\mathcal{Z}_k')\|^2/\|\mathbb{E}(\mathcal{E}_k\mathcal{E}_k')\|^2$ to be the ratio of the total variations of the signal and the noise components in the $k$th mode. Note that for a fixed mode $k=1,\,2,\,3$, all eigenvalues of $\mathbb{E}(\mathcal{E}_k\mathcal{E}_k')$ are equal. The table shows that in all settings the SNR is quite small and the simulation settings are quite challenging.






\begin{table}[ht]
\centering
\caption{Eigenvalues of $\mathbb{E}(\mathcal{E}_k\mathcal{E}_k')$ and SNR for Model~\eqref{eq:model_simus}, $k=1,\,2,\,3$.}
\label{tab:tab23}
\begin{tabular}{r|rrr|rrr}
\hline
&\multicolumn{3}{c}{Eigenvalues}&\multicolumn{3}{c}{SNR} \\

\cline{2-4}\cline{5-7}

$\sigma^2$ & Mode 1 & Mode 2 & Mode 3  & Mode 1 & Mode 2 & Mode 3 \\

\cline{1-4}\cline{5-7}

$0.1$      & 30    & 10  &   7.5 & 0.162 &0.756& 0.650 \\
$0.5$      & 150        & 50    &  37.5 & 0.006 &0.030 &0.026 \\
$1.0$       & 300     & 100    & 75.0 &  0.002 &0.007 &0.006 \\
\hline
\end{tabular}
\end{table}

For each of the $3$ variance values, we simulate 1000 data sets of size $n=1000$ from Model~\eqref{eq:model_simus}, where in each simulated data sample $\textbf{V}_i$, $i=1,2,3$, are randomly generated orthogonal matrices. The invertible matrices specifying the covariance structure of the core are of the form $\textbf{A}_i=\textbf{W}_i\textbf{D}_i\textbf{W}_i'$, where $\textbf{W}_i\in\mathbb{R}^{d_i\times d_i}$, $i=1,\dots,3$, are randomly generated orthogonal matrices and $\textbf{D}_1\approx\rm{diag}($1.857, 2.785, 3.714$)$, $\textbf{D}_2\approx\rm{diag}($1.797, 1.887, 2.247, 2.427, 2.696$)$ and $\textbf{D}_3\approx1.282\,\rm{diag}($1.05,1.1,\dots,1.45$)$. Furthermore, the mixing matrices $\textbf{U}_i\in\mathbb{R}^{p_i\times d_i}$ are taken to be the first $d_i$ columns of randomly generated orthogonal matrices in $\mathbb{R}^{p_i\times p_i}$, $i=1,2,3$.

The augmentation estimator has two tuning parameters, $s_k$ and $r_k$, for the estimation of the latent dimensions $d_1,\,d_2,\,d_3$. Additionally, one also has to choose the used estimator of the noise variance.
It is natural to choose the number of replications $s_k>0$ to be large to reduce variation in the results. Based on the simulation results of \cite{RadojicicLietzenNordhausenVirta2021}, we choose $s_k=50$, $k=1,\,2,\,3$, as that seems a good compromise between stability and computation time and since the choice of $s_k$ was in \cite{RadojicicLietzenNordhausenVirta2021} deemed to be less crucial than the choice of $r_k$. 
In the vector case, \cite{LuoLi2021} propose to use  $r\approx p_1/5$ which in light of our results at the population level, see  Corollary~\ref{cor:cor_1},
seems insufficiently large. This effect is further visualized in  Figure~\ref{fig:fig_beta} which also shows how, in practice, the choice of $r_k$ might be guided by the difference between the signal dimension and the full data dimension. However, the optimal choice of $r_k$ should be investigated further, most likely in the high-dimensional framework, and is beyond the scope of this paper.




To cover a wide range of values, we consider in the simulation study for each mode the values $r_k\in\{1,5,10,25,50\}$.
For the estimation of the noise variance in the matrix case, \cite{RadojicicLietzenNordhausenVirta2021} considered different estimators with the ``largest'' being the median estimator which also turned out to give the best performance. To evaluate if this is the case in the current scenario as well, we consider the quantile-based estimators, $\hat{\sigma}_{k,q}^2$ (see Lemma~\ref{lemma:lemma_consistency_of_estimators}), for the quantiles $q\in\{0,0.1,\dots,0.6\}$.
As a competitor we use the bootstrap-based ladle estimator  described in  Algorithm~\ref{alg::boot}, with the number of bootstrap samples $s_k=200$ and where we ignored the if-condition $p_k \leq 10$ and always used $q_k = p_k-1$.
The results of the simulation study are presented in
Figure~\ref{fig:fig_simus}, where in each sub-figure, columns correspond to the row-dimensions $r_k$, $k=1,2,3$, of the augmentation matrices, while the rows specify the standardized noise variance.

\begin{figure}
    \centering
    \includegraphics[width=0.55\linewidth]{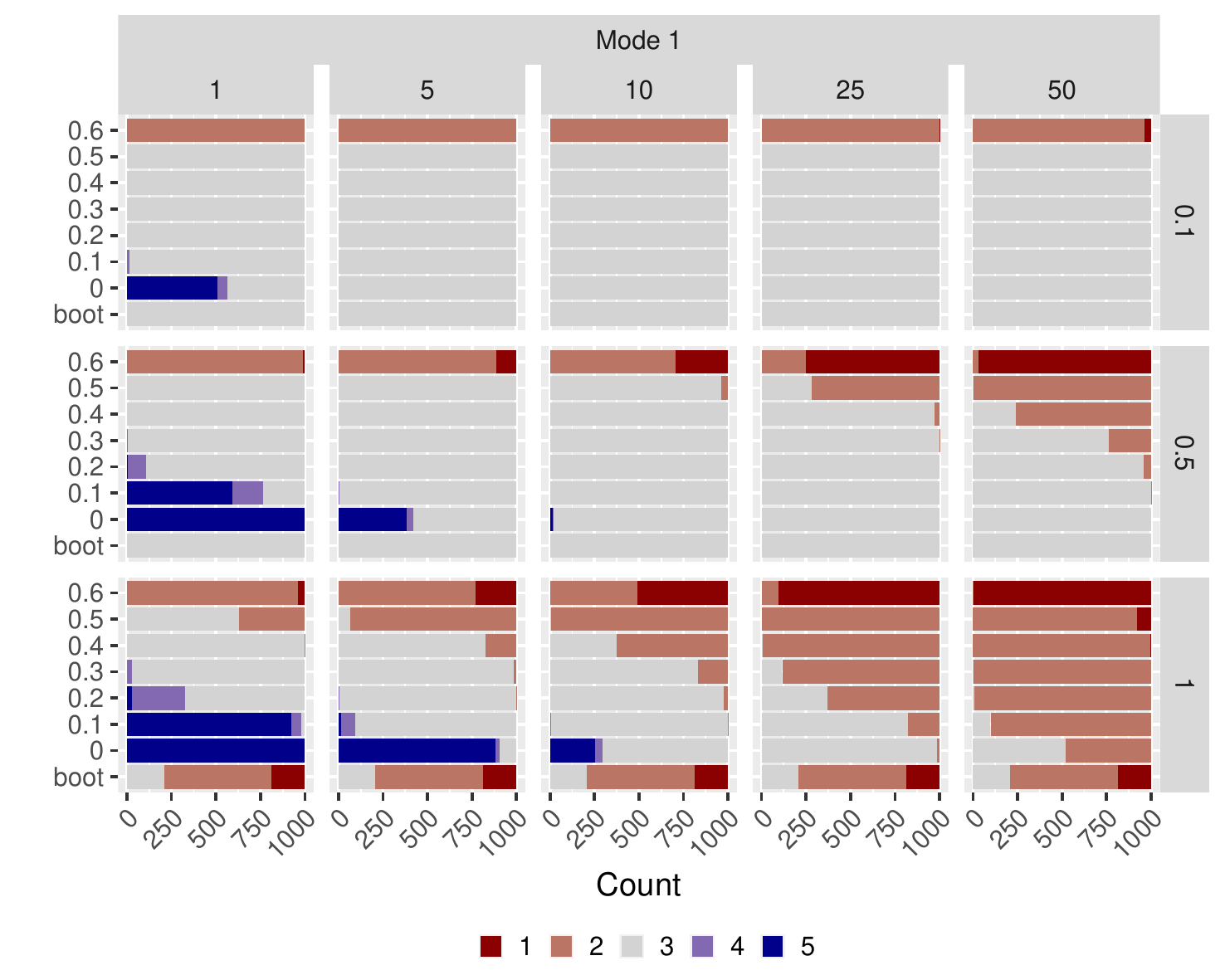}
    \includegraphics[width=0.55\linewidth]{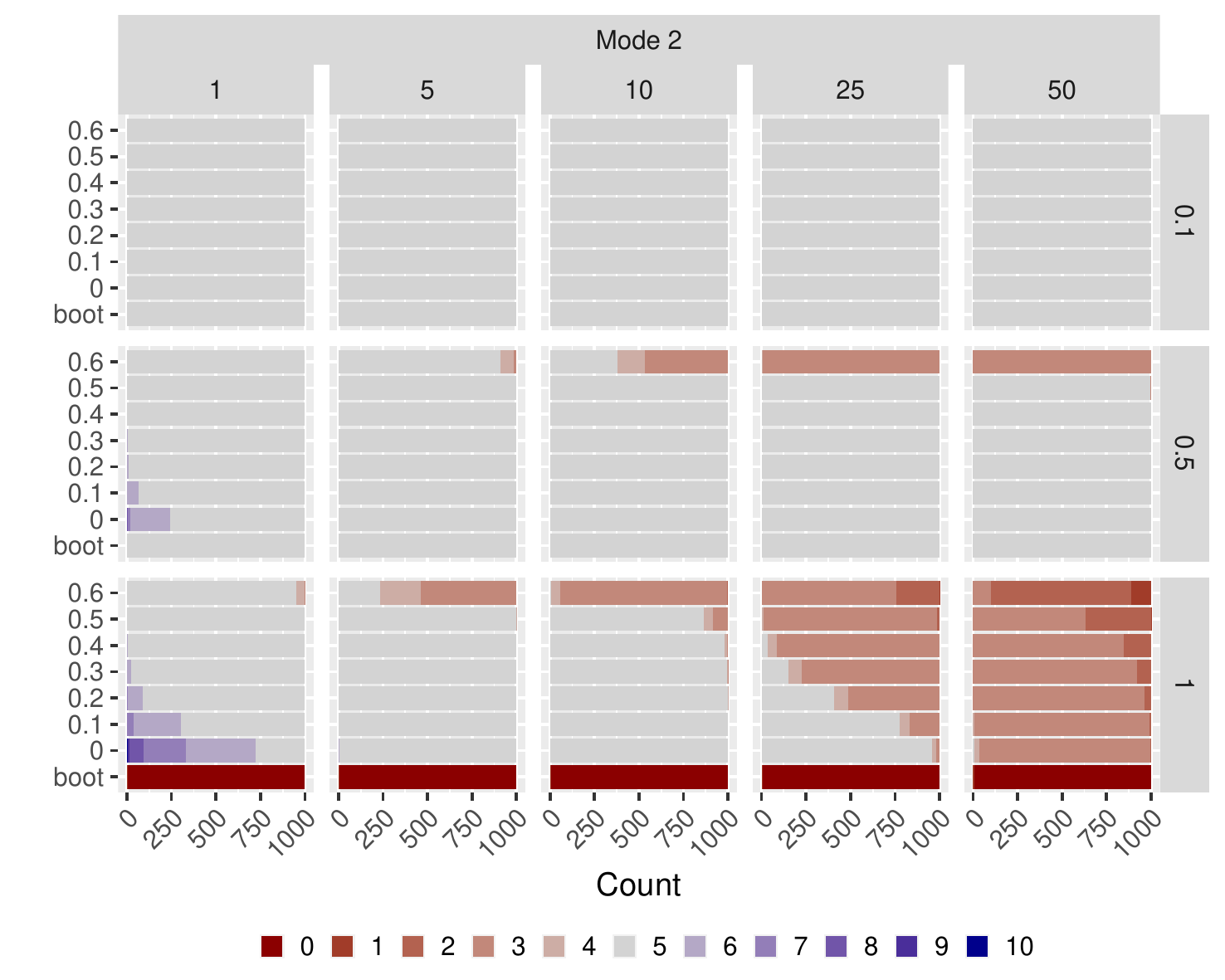}
    \includegraphics[width=0.55\linewidth]{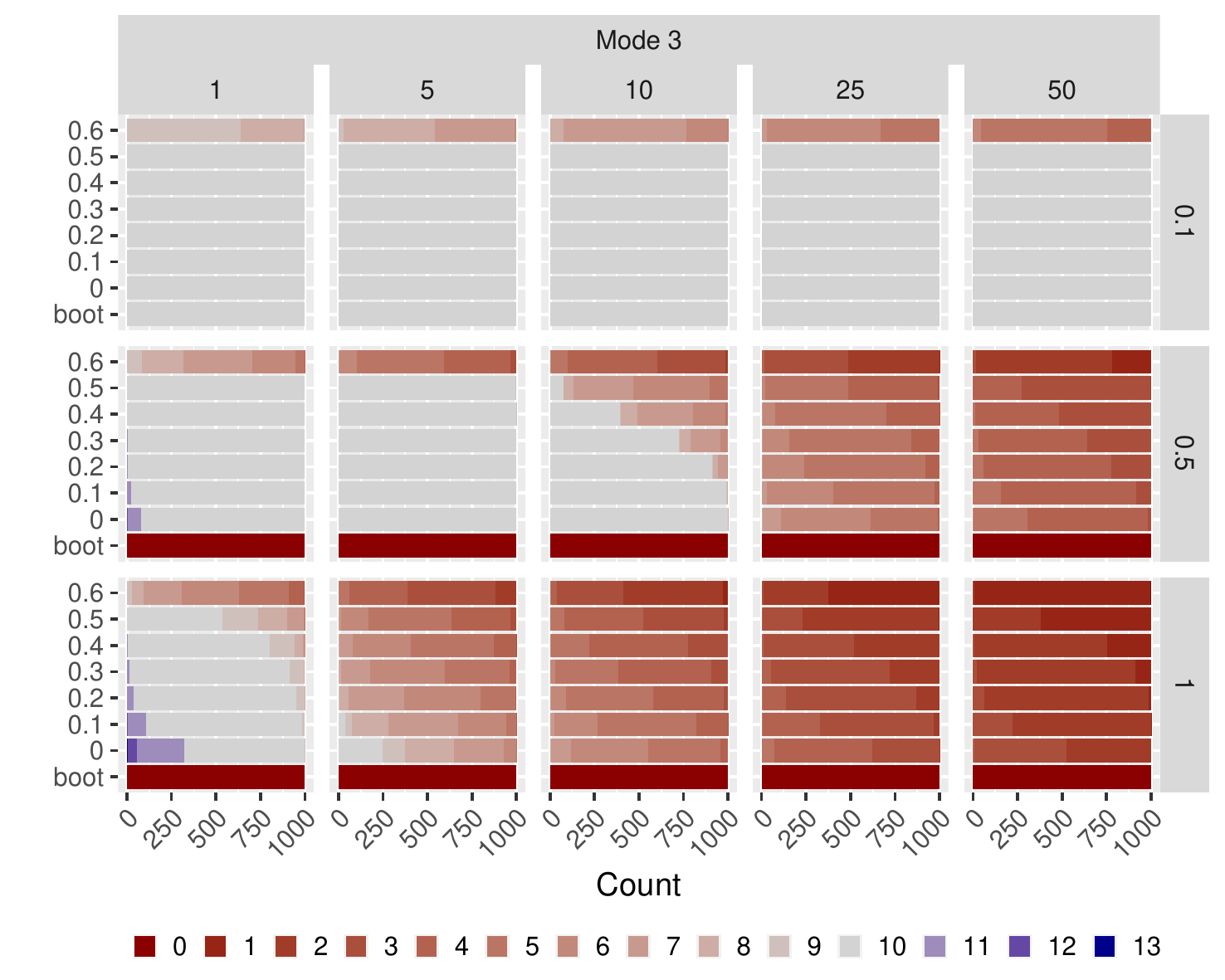}
    \caption{Frequencies of estimated latent dimensions in Model~\ref{eq:model_simus} based on 1000 repetitions. The true latent dimensions are $d_1=3$, $d_2=5$ and $d_3=10$ and are always marked as grey. Note that for ladle (method boot) the estimates are the same for all $r_k$.}
    \label{fig:fig_simus}
\end{figure}

The results show that, for the smallest value of the noise variance $\sigma^2 = 0.1$, the correct signal dimension is obtained in all modes and for all values of $r_k$ if the quantile index used for the noise variance estimation is not too extreme. However, when the SNR is decreased the number of augmented components $r_k$ and the quantile used for the noise variance estimation become of more relevance, especially for the modes with large values of $p_k$. Counter-intuitively to our theory, too large values of $r_k$ and noise variance lead to underestimation of the signal dimension. Recall, however, that in these simulations the SNR is in all cases quite low showing that the estimators perform all rather well. Nevertheless, we suggest not to use too drastic values for $r_k$, restricting to values such as $r_k=10$ or $r_k=25$ and using 0.2 or 0.3 as the quantile level for the noise variance estimation. The simulations also show that if $p_k$ is small and the SNR is not too low, the ladle works very well, but it also seems to suffer the most from deviations from these optimal settings and then always estimates the signal dimension as zero.

The disagreement between Corollary~\ref{cor:cor_1} and the simulation results with respect to increasing the number of augmented components especially when $p_k$ is large might be due to the fact that our asymptotic results are stated for fixed dimensionality and growing sample size $n$. Whereas, in practice, $n$ is fixed and, therefore, if $p_k + r_k$ is relatively large we might fall into the area of the Marchenko–Pastur law \cite{GotzeTikhomirov2004}. This will be investigated further in future research. 


\subsection{Example}\label{sec:real_data_example}
To illustrate the augmention estimator we considered 882
$224 \times 224$ RGB images of butterflies from various species, mostly taken in natural habitats. We assume that the SNR of the data is larger than in the previous simulation and therefore use $s_k=50$, $r_k=5$ and estimate the noise variance using the 30\% quantile.

Based on these values, Figure~\ref{fig:log_example} visualizes the
graph of our augmentation objective function $\hat{g}_k$ in \eqref{eq:phi_aug} for each mode on a logarithmic scale. The estimated signal dimension, obtained as the modewise minima, is (103, 111, 3). The fact that $d_3$ is estimated to be $3$ indicates that there is a lot of information in the data contained in all three colors, a fact that indeed is plausible, since  butterflies are greatly characterized by the color of their wings.
Figure~\ref{fig:example} illustrates five selected original (top) and reconstructed (bottom) images of butterflies and differences between the two are indeed only visible when zooming in.


\begin{figure}[ht]
\centering
\includegraphics[width = 0.6 \linewidth]{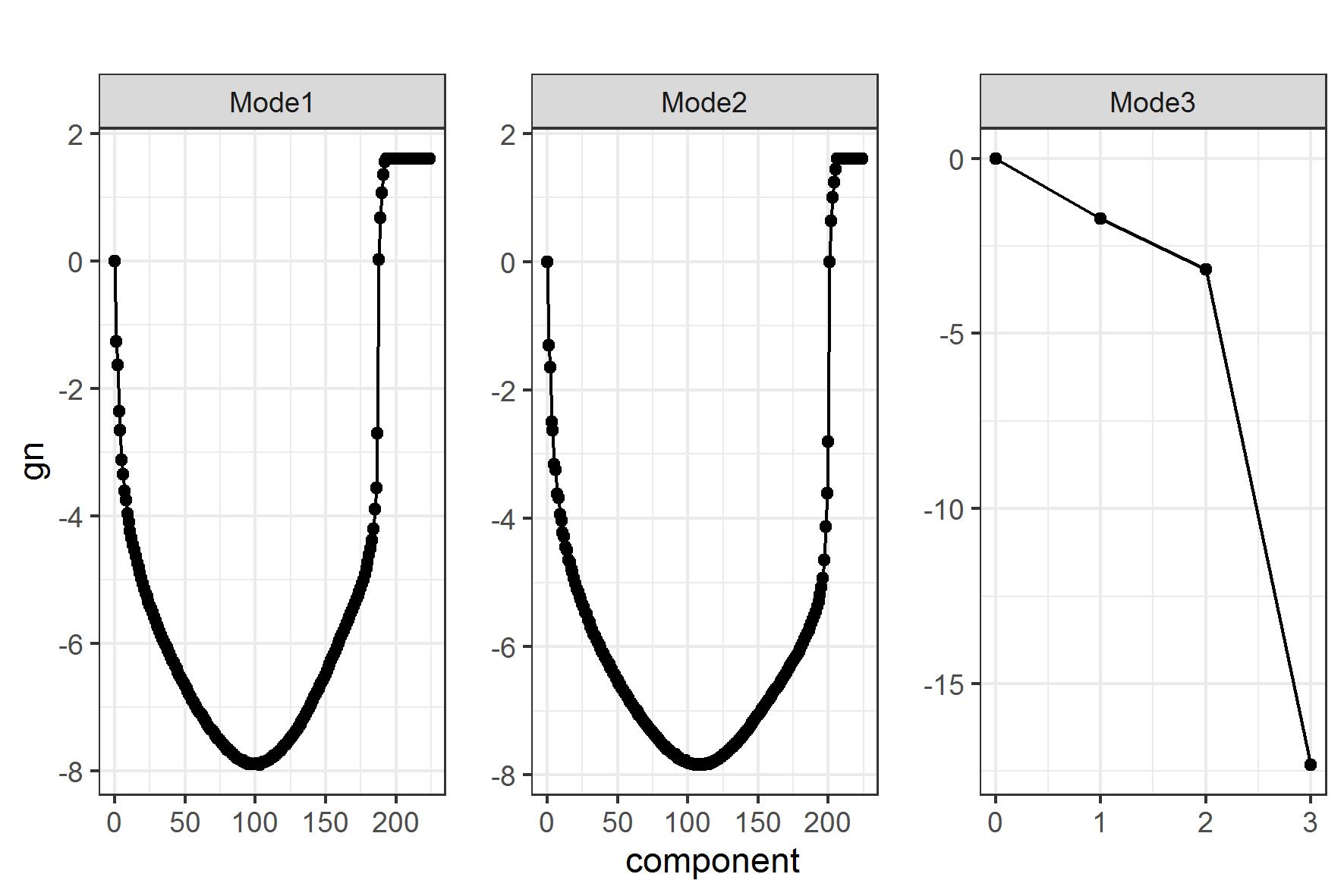}
\caption{Logarithmized objective function $\hat{g}_k$ for the augmentation estimator using $r_k = 25$, $s_k=20$,   $k=1,2,3$, calculated for the $\textit{butterfly}$ data set.}
\label{fig:log_example}
\end{figure}

\section{Discussion}\label{sec:discussion}

Data objects which have a natural representation as tensors, like images or video, are increasingly common and often the dimensions of these tensors are huge. In such cases it is often assumed that the data contain a lot of noise and that a representation using smaller tensors should be sufficient to capture the information content of the data. A difficult key question is then to decide on the dimensions of these smaller signal tensors. We considered this problem in the framework of a tensorial principal component analysis, also known as HOSVD, and suggested an automated procedure for the order determination by extending the works of \cite{LuoLi2021,RadojicicLietzenNordhausenVirta2021}. The properties of the novel estimator were rigorously derived and its usefulness was demonstrated using simulated and real data.

We also observed that, for finite $n$, increasing the number $r_k$ of augmented rows did not fully agree with the theory, which might be related to the Marchenko-Pastur law. This will be investigated further in future research by assuming a high-dimensional framework where we will also consider the use of different norms when computing the objective criterion. Furthermore, we plan to derive also a hypothesis test which would allow inference about the signal tensor dimension.


\bibliographystyle{myIEEEtran}
\bibliography{AugPCA}

\newpage

\appendix

\section{Tensor notation}\label{sec:tensor}

Let $\mathcal{A} = (a_{i_1, \dots, i_m})\in\mathbb{R}^{p_1\times\cdots\times p_m}$ be a tensor of order $m$. The $m$ ``directions''
from which we can look at $\mathcal{A}$ are called the modes of $\mathcal{A}$. As discussed, the number of elements in a high-order tensor is in general rather large, and it is therefore often useful to split a tensor into smaller pieces using $k$-mode vectors. For a fixed mode $k$, a single $k$-mode vector of a tensor $\mathcal{A}$ is obtained by letting the $k$th index of $a_{i_1, \dots, i_m}$ vary while simultaneously keeping the rest fixed. The total number of $k$-mode vectors is thus $\rho_k := \prod_{i\neq k}p_i$ and each of them is $p_k$-dimensional. The $p_k \times \rho_k$-dimensional matrix $\mathcal{A}_k$ having all $k$-mode vectors of $\mathcal{A}$ as its columns is known as the $k$-unfolding/flattening/matricization of the tensor $\mathcal{A}$.  The ordering of the $k$-mode vectors in $\mathcal{A}_k$  is for our purposes irrelevant as long as it is consistent, and we choose the cyclical
ordering as suggested in~\cite{Lathauwer2000}. For a matrix $\textbf{A}_m\in\mathbb{R}^{q_k\times p_k}$, the $k$-mode multiplication $\mathcal{A}\times_m\textbf{A}_m$ of a tensor $\mathcal{A}$ by the matrix $\textbf{A}_m$ is defined as the tensor of order $p_1\times\cdots\times p_{k-1}\times q_k\times p_{k+1}\times\cdots\times p_m$ obtained by pre-multiplying each $k$-mode vector of $\mathcal{A}$ by $\textbf{A}_k$. Often, linear transformations are applied simultaneously from all $m$ modes and we therefore use the notation $\mathcal{A}\times_{i=1}^m\textbf{A}_i:=\mathcal{A}\times_1\textbf{A}_1\times_2\cdots\times_m\textbf{A}_m$. The flattenings of such linear transformation have a particularly nice form. Namely, the $k$-flattening of $\mathcal{A}\times_{i=1}^m\textbf{A}_i$ is $\textbf{A}_k \mathcal{A}_k (\textbf{A}_{k+1}\otimes\cdots\otimes\textbf{A}_m\otimes\textbf{A}_1\otimes\cdots\otimes\textbf{A}_{k-1})'$.
Relevant to our purposes is also the Frobenius norm $\|\mathcal{A}\|_\mathrm{F}$ of  a tensor $\mathcal{A}$, defined as the square root of the sum of squared elements of $\mathcal{A}$.  For more details on tensor algebra and the corresponding notation see e.g.~\cite{Lathauwer2000,Kolda2009}.

\section{Proofs of technical results}\label{sec:proofs}

\begin{proof}[Proof of Lemma \ref{lemma:lemma1}]
$\mathcal{E}$ has spherical distribution, implying that, for any orthogonal matrices $\textbf{V}_i\in\mathbb{R}^{p_i\times p_i}$, $i=1,\dots,m$, $\mathcal{E}\times_{i=1}^m\textbf{V}_i\sim \mathcal{E}$. Then
$$
\textbf{V}_k\mathcal{E}_k(\textbf{V}_{k+1}\otimes\textbf{V}_{k+2}\otimes\cdots\otimes\textbf{V}_m\otimes\textbf{V}_1\otimes\textbf{V}_2\otimes\cdots\otimes\textbf{V}_{k-1})'\sim\mathcal{E}_k,
$$
for all orthogonal matrices $\textbf{V}_i$. Since the Kronecker product of orthogonal matrices is again an orthogonal matrix \cite{laub2005}, then $(\textbf{V}_{k+1}\otimes\textbf{V}_{k+2}\otimes\cdots\otimes\textbf{V}_m\otimes\textbf{V}_1\otimes\textbf{V}_2\otimes\cdots\otimes\textbf{V}_{k-1})$ is a $(\prod_{i\neq k}p_i\times\prod_{i\neq k}p_i)$ orthogonal matrix. Thus, $\mathbb{E}(\mathcal{E}_k\mathcal{E}_k')=\textbf{V}_k\mathbb{E}(\mathcal{E}_k\mathcal{E}_k')\textbf{V}_k'$, for all $(p_k\times p_k)$ orthogonal matrices $\textbf{V}_k$, further implying that $\mathbb{E}(\mathcal{E}_k\mathcal{E}_k')=\sigma_k^2\textbf{I}_{p_k}$ for some $\sigma_k^2>0$.
\end{proof}

\begin{proof}[Proof of Theorem \ref{thm:thm0}]
Let $\boldsymbol{\beta}_{}^*=(\boldsymbol{\beta}_{},\boldsymbol{\beta}_{S}) = \textbf{B}^*_{k,0}\textbf{a} \in\mathbb{R}^{p_k+r_k}$, where $\textbf{a}=(\textbf{a}_1',\textbf{a}_2')'$ follows a uniform distribution on the unit sphere in $\mathbb{R}^{p_k-d_k+r_k}$ and $\textbf{a}_1\in\mathbb{R}^{p_k-d_k}$, $\textbf{a}_2\in\mathbb{R}^{r_k}$.

Due to the sphericality of $\textbf{a}$, we can take $\textbf{B}_{k,0}^*=\mathrm{diag}([\textbf{u}_{d_k+1},\dots,\textbf{u}_{p_k}],\textbf{I}_{r_k})$, where  $\{\textbf{u}_{d_k+1},\dots,\textbf{u}_{p_k}\}$ is a fixed orthonormal basis of the null space of $\textbf{U}_k$. Thus, the norm of the augmented part is $\|\boldsymbol\beta_S\|=\|\textbf{a}_2\|$. This now shows that $\|\boldsymbol\beta_{S}\|^2$, as the squared norm of the $r_k$-dimensional sub-vector of the random vector $\textbf{a}$ with a uniform distribution on the unit sphere in $\mathbb{R}^{p_k-d_k+r_k}$, has the $\mathrm{Beta}(r_k/2,(p_k-d_k)/2)$-distribution.
\end{proof}

In the following proofs we adopt the notation that  $\mathcal{X}^1,\dots,\mathcal{X}^n$ is an i.i.d. sample from Model~\ref{def:tensor_model}. For simplicity of notation, we assume without loss of generality that the mean $\mathcal{M}$ of $\mathcal{X}$ is a zero-tensor. Before giving proofs of Corollary~\ref{cor:cor_2} and Lemma~\ref{lemma:lemma_consistency_of_estimators}, we present first an auxiliary result.

\begin{lemma}\label{lemma:lemma5}
Let $\textbf{M}_k^*$ and $\hat{\textbf{M}}_k^*$ be as defined in the main text. Then,  $\|\hat{\textbf{M}}_k^*-\textbf{M}_k^*\|\to_P 0$.
\end{lemma}

\begin{proof}[Proof of Lemma \ref{lemma:lemma5}]
It is straightforward to verify that
\begin{align*}
    \textbf{M}_k^* &= \begin{pmatrix}
\textbf{U}_k\mathbb{E}(\mathcal{Z}_k\mathcal{Z}_k')\textbf{U}_k' &\textbf{0}\\
\textbf{0} &\textbf{0}
\end{pmatrix}, \\
\hat{\textbf{M}}_k^* &= \frac{1}{n}\sum_{i=1}^n\begin{pmatrix}
\mathcal{X}_{k,i}\mathcal{X}_{k,i}' &\hat\sigma_k\mathcal{X}_{k,i}\textbf{X}_{S,i}'\\
\hat\sigma_k\textbf{X}_{S,i}\mathcal{X}_{k, i}' & \hat\sigma_k^2\textbf{X}_{S,i}\textbf{X}_{S,i}'
\end{pmatrix}-\hat\sigma_k^2\textbf{I}_{p_k+r_k}.
\end{align*}
Observe now the first diagonal block of $\hat{\textbf{M}}_k^*$. Due to WLLN, continuous mapping theorem and the fact that $\hat\sigma_k$ is a consistent estimator of $\sigma_k$, we have $\frac{1}{n}\sum_{i=1}^n\mathcal{X}_{k,i}\mathcal{X}_{k,i}'-\hat\sigma_k^2\textbf{I}_{p_k}\to_P \mathbb{E}(\mathcal{X}_k\mathcal{X}_k')-\sigma_k^2\textbf{I}_{p_k}=\textbf{U}_k\mathbb{E}(\mathcal{Z}_k\mathcal{Z}_k')\textbf{U}_k'$. The convergence of the second diagonal block is proved similarly. Furthermore, by assumption,
$\hat\sigma_k^2\to_{P}\sigma_k^2$ and $\frac{1}{n}\sum_{i=1}^n\textbf{X}_{S,i}\mathcal{X}_{k,i}'\to_{P}\mathbb{E}(\textbf{X}_S \mathcal{X}_k')=\textbf{0}$, implying the convergence of the off-diagonal blocks. Finally, since matrix norms are continuous functions, we obtain that $\|\hat{\textbf{M}}_k^*-\textbf{M}_k^*\|\to_P 0$.
\end{proof}

\begin{proof}[Proof of Corollary~\ref{cor:cor_2}]
$\hat{\textbf{M}}_k^*$, $\textbf{M}_k^*$  and their difference $\textbf{R}_{k,n}=\hat{\textbf{M}}_k^*-\textbf{M}_k^*$ are symmetric matrices. Weyl's theorem then implies that, for any $i\in \{1,\dots , p_k+r_k\}$, we have $0\leq |\hat\lambda_{k,i}-\lambda_{k,i}|\leq \|\textbf{R}_{k,n}\|_2=o_P(1)$, where $\|\textbf{R}_{k,n}\|_2=o_P(1)$ holds due to Lemma~\ref{lemma:lemma5}.
\end{proof}

\begin{proof}[Proof of Lemma \ref{lemma:lemma_consistency_of_estimators}]
For simplicity of notation, we do not write the mean terms $\mathcal{M}_k$ and $\bar{\mathcal{X}}_k$ below. Mimicking the proof of Lemma~\ref{lemma:lemma5} it is easy to show that
$\left\|\frac{1}{n}\sum_{i=1}^n(\mathcal{X}_{k,i}{\mathcal{X}_{k,i}}')-\mathbb{E}\left(\mathcal{X}_{k,i}\mathcal{X}_{k,i}'\right)\right\|^2=o_P(1)$, for $k=1,\dots,m$. Mimicking further the proof of Corollary~\ref{cor:cor_2} for $\mathbb{E}(\mathcal{X}_{k,i}\mathcal{X}_{k,i}')$ one can show that the eigenvalues satisfy $\hat\sigma_{k,i}^2=\sigma_{k,i}^2+o_P(1)$, $k=1,\dots,m$. We then have that $\mathbb{E}(\mathcal{X}_k\mathcal{X}_k')=\textbf{U}_k'\mathbb{E}(\mathcal{Z}_k\mathcal{Z}_k')\textbf{U}_k+\sigma_k^2\textbf{I}_{p_k}$ implying that $\sigma_{k,i}^2=\lambda_{k,i}+\sigma_k^2$, $k=1,\dots,m$. Moreover,
\begin{align*}
\hat S_k =\{\frac{p_i}{p_k}\hat\sigma_{i,j_i}^2:i=1,\dots,m,\,j_i=1,\dots,d_i\}\cup\hat{S}_{0k},
\end{align*}
where
\begin{align*}
\hat S_{0k}=\{\frac{p_i}{p_k}\hat\sigma_{i,j_i}^2:i=1,\dots,m,\,j_i=d_i+1,\dots,p_i\}
\end{align*}
and each element of $\hat S_{0k}$ is due to the upper discussion of form $\sigma_k^2+o_P(1)$.
\begin{itemize}
    \item[i)] If $d_1+\dots+d_m<(1-q)(p_1+\dots+p_m)$, then both estimators $\hat\sigma^2_{k,q_1}$ and $\bar\sigma_{k,q_1}^2$, for $q_1\leq q$ belong to $\hat S_{0k}$, which proves the first statement.
    \item[ii)] If $d_1+\cdots+d_m<p_1+\cdots+p_m$ then the estimator $\min\{\hat S_k\}$ belongs to $\hat S_{0k}$, proving the second statement.
\end{itemize}
\end{proof}

In the following, we let  $\hat\sigma_k^2=\hat{\sigma}_k^2(\mathcal{X}^1,\dots,\mathcal{X}^n)$ be a consistent estimator of the error variance in the $k$th mode.

\begin{proof}[Proof of Theorem \ref{thm:consistency_eigenvectors}]
Lemma~\ref{lemma:lemma5} gives that $\|\hat{\textbf{M}}_k^*-\textbf{M}_k^*\|=o_P(1)$, thus implying that each element of the difference of those two matrices is $o_P(1)$. For simplicity, we write $\hat{\textbf{M}}_k^*-\textbf{M}_k^*=o_P(1)$, where the convergence is element-wise in probability. Lemma~\ref{lemma:lemma5} and Corollary~\ref{cor:cor_2} now imply
\begin{align}\label{eq:eq_1}
\sum_{i=1}^{d_k}{\lambda}_{k,i}^2\hat{\boldsymbol{\beta}}_{k,i}^*{\hat{\boldsymbol{\beta}}_{k,i}^*}{}'-\sum_{i=1}^{d_k}{\lambda}_{k,i}^2{\boldsymbol{\beta}}_{k,i}^*{\boldsymbol{\beta}_{k,i}^*}'=o_P(1).
\end{align}
If we multiply~\eqref{eq:eq_1} by $\hat{\boldsymbol{\beta}}_{k,1}^*$ from the right-hand side, we obtain
$$
\hat{\boldsymbol{\beta}}_{k,1,S}=\frac{1}{\lambda_{k,1}^2}\sum_{i=1}^{d_k}{\lambda}_{k,i}^2({\boldsymbol{\beta}_{k,i}^*}'\hat{\boldsymbol{\beta}}_{k,1}^*){\boldsymbol{\beta}}_{k,i,S}+o_P(1)
=o_P(1),
$$
where the last equality holds due to $\boldsymbol\beta_{k,i,S}=\textbf{0}$ as shown in Section~\ref{sec:augmentation}. We repeat the same procedure for all $\hat{\boldsymbol{\beta}}_{k,i,S}$, $i=1,\dots,d_k$, thus proving statement (i).

We next prove statement (ii). If we further multiply~\eqref{eq:eq_1} from both sides by any  $\boldsymbol{\beta}_k^*\in\mathrm{Ker}(\textbf{M}_k^*)$, we
obtain $\sum_{i=1}^{d_k}\lambda_{k,i}^2({\hat{\boldsymbol\beta}_{k,i}^*}{}'\boldsymbol\beta_k^*)^2=o_P(1)$, thus implying that
\begin{equation}\label{eq:alpha_o_p(1)}
{\hat{\boldsymbol\beta}_{k,i}^*}{}'\boldsymbol\beta_{k,j}^*=o_P(1),\,\, i=1,\dots,d_k,\,\, j=d_k+1,\dots,p_k + r_k.
\end{equation}

 We let $\hat{\textbf{M}}^*_{k,0}$ denote the matrix that is otherwise exactly as  $\hat{\textbf{M}}^*_{k}$ but with every instance of $\hat{\sigma}^2_k$ replaced with $\sigma^2_k$ and with the sample means $\bar{\textbf{X}}^*_k$ replaced with zero matrices. That is, in $\hat{\textbf{M}}^*_{k,0}$ the augmented part is scaled with the true value $\sigma_k$ instead of the estimator $\hat{\sigma}_k$ and we ``center'' using the population-level matrix $\sigma_k^2 \textbf{I}_{p_k + r_k}$. By Lemma \ref{lemma:lemma_consistency_of_estimators} and the law of large numbers, we then have $ \hat{\textbf{M}}^*_{k} = \hat{\textbf{M}}^*_{k,0} + o_p(1) $.

Define then the  following $(p_k + r_k) \times (p_k + r_k)$ orthogonal matrix,
\begin{align*}
    \textbf{W} := \begin{pmatrix}
    \textbf{U}_k & \textbf{U}_0 & \textbf{0} \\
    \textbf{0} & \textbf{0} & \textbf{I}_{r}
    \end{pmatrix},
\end{align*}
where the $p_k \times d_k$ matrix $\textbf{U}_k$ is as in \eqref{def:matrix_k_flattening_model}, and the $p_k \times (p_k - d_k)$ matrix $\textbf{U}_0$ is an arbitrary matrix that makes $(\textbf{U}_k, \textbf{U}_0)$ orthogonal. Block-wise multiplication now shows that $\textbf{W}' ((\mathcal{X}_{k,i})',\sigma_k\textbf{X}_{S,i}')'$ equals
\begin{align}\label{eq:principal_component_decomposition}
    \begin{pmatrix}
    \mathcal{Z}_{k,i} (\textbf{U}_{-k}^\otimes)' + \textbf{U}_k' \mathcal{E}_{k,i} \\
    \textbf{U}_0' \mathcal{E}_{k,i} \\
    \sigma_k \textbf{X}_{S,i}
    \end{pmatrix}.
\end{align}
The final $p_k - d_k + r_k$ rows of the matrix in \eqref{eq:principal_component_decomposition} consist solely of mutually independent elements from the distribution $\mathcal{N}(0, \sigma^2_k)$. As these rows are further independent from the first $d_k$ rows, the distribution of the matrix \eqref{eq:principal_component_decomposition} is invariant to multiplication from the left with matrices of the form $\mathrm{diag}(\textbf{I}_{d_k}, \textbf{V}')$ where $\textbf{V}$ is an arbitrary $(p_k - d_k + r_k) \times (p_k - d_k + r_k)$ orthogonal matrix. Consequently, we also have the distributional invariance,
\begin{align}\label{eq:distribution_invariance}
    \textbf{W}' \hat{\textbf{M}}^*_{k,0} \textbf{W} \sim \mathrm{diag}(\textbf{I}_{d_k}, \textbf{V}') \textbf{W}' \hat{\textbf{M}}^*_{k,0} \textbf{W} \mathrm{diag}(\textbf{I}_{d_k}, \textbf{V}),
\end{align}
for any orthogonal $\textbf{V}$. As the distribution of $\mathcal{E}_{k,i}$ is absolutely continuous, we have, for large enough $n$, that the eigenvalues of $\hat{\textbf{M}}^*_{k,0}$ are almost surely simple. Without loss of generality, we next restrict solely to this event. Consequently, the corresponding eigenvectors $\hat{\boldsymbol{\gamma}}_{k,j}^*$, $j = 1, \ldots, p_k + r_k$ are uniquely defined up to their sign, which we choose arbitrarily. Next, for a symmetric random matrix $\textbf{A}$ with simple eigenvalues, we let $\boldsymbol{\gamma}_j(\textbf{A})$ denote its $j$th eigenvector multiplied with a uniformly random sign (the sign multiplication guarantees that the distribution of $\boldsymbol{\gamma}_j(\textbf{A})$ is well-defined even though the eigenvectors themselves are unique only up to sign.

Applying the map $\textbf{A} \mapsto \boldsymbol{\gamma}_k(\textbf{A})$ to relation \eqref{eq:distribution_invariance} gives that
\begin{align*}
    s_1 \textbf{W}' \hat{\boldsymbol{\gamma}}_{k,j}^* \sim s_2
    \mathrm{diag}(\textbf{I}_{d_k}, \textbf{V}') \textbf{W}' \hat{\boldsymbol{\gamma}}_{k,j}^*,
\end{align*}
where $s_1, s_2$ are uniformly random signs. Thus, the sub-vector consisting of the final $p_k - d_k + r_k$ elements of $\textbf{W}' \hat{\boldsymbol{\gamma}}_{k,j}^*$ is orthogonally invariant in distribution. Using the decomposition $ \hat{\boldsymbol{\gamma}}_{k,j}^* = (\hat{\boldsymbol{\gamma}}_{k,j,1}, \hat{\boldsymbol{\gamma}}_{k,j,S}) $, when $j > d_k$ this sub-vector writes $ \hat{\textbf{m}} := (\textbf{U}_0' \hat{\boldsymbol{\gamma}}_{k,j,1}, \hat{\boldsymbol{\gamma}}_{k,j,S})$. By mimicking \eqref{eq:alpha_o_p(1)}, we see that $ \textbf{U}_k' \hat{\boldsymbol{\gamma}}_{k,j,1} = o_p(1) $, and then the unit length of $\hat{\boldsymbol{\gamma}}_{k,j}^*$ guarantees that $\| \hat{\textbf{m}} \| = 1 + o_p(1)$. By writing
\begin{align*}
    \hat{\textbf{m}} = \hat{\textbf{m}}/\| \hat{\textbf{m}} \| + (1 - 1/\|\hat{\textbf{m}}\|) \hat{\textbf{m}} = \hat{\textbf{m}}/\| \hat{\textbf{m}} \| + o_p(1),
\end{align*}
we see that the limiting distribution of $\hat{\textbf{m}}$ is the uniform distribution on the unit sphere in $\mathbb{R}^{p_k - d_k + r_k}$. As such, the limiting distribution of $\| \hat{\boldsymbol{\gamma}}_{k,j,S} \|^2$ is $\mathrm{Beta}\{ r_k/2, (p_k - d_k)/2 \}$. As $ \hat{\textbf{M}}^*_{k} = \hat{\textbf{M}}^*_{k,0} + o_p(1) $ and as the extraction of the $j$th eigenvector is (up to sign) a continuous mapping when the $j$th eigenvalue is simple, the same limiting distribution is valid also for $\| \hat{\boldsymbol\beta}_{k,j,S} \|^2$, concluding the proof.
\end{proof}

\begin{proof}[Proof of Corollary~\ref{cor:cor_3}]
Let $F_n$ and $F$ be the CDFs of $\|\hat{\boldsymbol{\beta}}_{k,i,S}\|^2$ and $\mathrm{Beta}\{ r_k/2, (p_k - d_k)/2 \}$, respectively, and observe that $F$ is continuous at $0$ (from the right). Furthermore, due to Theorem~\ref{thm:consistency_eigenvectors}, $\|\hat{\boldsymbol{\beta}}_{k,i,S}\|^2\rightsquigarrow \mathrm{Beta}\{ r_k/2, (p_k - d_k)/2 \}$, $i > d_k$.

Observe then that, since $F_n$ is a CDF, it is a bounded and increasing function. Thus, for fixed $n\in\mathbb{N}$, $\lim_{\varepsilon\to 0}F_n(\varepsilon)$ exists. Furthermore, as $\|\hat{\boldsymbol\beta}_{k,i,S}\|^2$ converges weakly to $\mathrm{Beta}\{ r_k/2, (p_k - d_k)/2 \}$, and $F$ is continuous on $[0,1]$, then $F_n$ converges uniformly to $F$ on $[0,1]$. Therefore, due to Moore-Osgood theorem the following limits interchange:
$$
\lim_{n\to\infty}\lim_{\varepsilon\to 0}F_n(\varepsilon)=\lim_{\varepsilon\to 0}\lim_{n\to\infty}F_n(\varepsilon)=\lim_{\varepsilon\to 0}F(\varepsilon)=0,
$$
where the last equality holds since $F$ is continuous at $0$ from the right as the CFD of the $\mathrm{Beta}\{ r_k/2, (p_k - d_k)/2 \}$-distribution.
\end{proof}

\begin{proof}[Proof of Theorem \ref{thm:consistency_of_d}]
We give a proof for $s=1$. The general case is proven using the same technique.

For $i<d_k$, the denominator of $\hat{\Phi}_k$ satisfies
\begin{align*}
1+\sum_{j=1}^i\hat\lambda_{k,j}\to_P 1+\sum_{j=1}^i\lambda_{k,i}>1.
\end{align*}
The continuous mapping theorem ($x\mapsto 1/x$ is continuous on $[1,\infty)$) then implies that
\begin{align*}
(1+\sum_{j=1}^i\hat\lambda_{k,j})^{-1}\to_P (1+\sum_{j=1}^i\lambda_{k,j})^{-1}\in (0,1).
\end{align*}
{Since $\hat\lambda_{k, i+1}\to_P\lambda_{k,i+1}>0$, then $\hat{\Phi}_k(i)
\to_P \lambda_{k,i+1}/(1+\sum_{j=1}^{i+1}\lambda_{k,j})>0$.}


This shows then that $\hat{g}_k(i)$, $i < d_k$, is bounded from below by a quantity which converges in probability to something strictly greater than zero.

For $i>d_k$, statement $(ii)$ of Corollary~\ref{cor:cor_3} gives that for any sequence $\varepsilon_n>0$ such that $\varepsilon_n \rightarrow 0$, we have  $\mathbb{P}(\|\hat{\boldsymbol\beta}_{k,i,S}\|^2>\varepsilon_n)\to 1$, as $n\to\infty$.  Thus, $\mathbb{P}(\hat{g}_k(i)>\varepsilon_n)\geq \mathbb{P}(\|\hat{\boldsymbol\beta}_{k,i,S}\|^2>\varepsilon_n)\to 1$, as $n\to\infty$.

For $i=d_k$, by Theorem~\ref{thm:consistency_eigenvectors} we have $\sum_{j=1}^{d_k}\|\hat{\boldsymbol\beta}_{k,j,S}\|^2=o_P(1)$, and by Corollary~\ref{cor:cor_2}, $\hat\lambda_{k,d_k+1}=\lambda_{k,d_k+1}+o_P(1)=o_P(1),$ implying that $\hat{\Phi}_k(d_k) \rightarrow_p 0 $. Thus $\hat{g}_k(d_k) = o_p(1) $.

The above behaviours of $g_k$ in the three different cases $i < d_k$, $i = d_k$, $i > d_k$ now together give the desired claim.

\end{proof}

\end{document}